\documentclass[twoside,10pt]{article} 
\usepackage{amsfonts}
\usepackage{amsmath}
\usepackage{amssymb}
\usepackage{epsfig}
\usepackage{graphicx}
\usepackage{latexsym}
\usepackage{subfigure}
\usepackage{times}
\usepackage{hyperref}

\newcommand{\app}{\textsc{Cogra}}
\newcommand{\greta}{{\small GRETA}}

\newcommand{\anysem}{{\small \textsf{ANY}}}
\newcommand{\nextsem}{{\small \textsf{NEXT}}}
\newcommand{\contsem}{{\small \textsf{CONT}}}

\newcommand{\return}{{\small \textsf{RETURN}}}
\newcommand{\pattern}{{\small \textsf{PATTERN}}}
\newcommand{\semantics}{{\small \textsf{SEMANTICS}}}
\newcommand{\where}{{\small \textsf{WHERE}}}
\newcommand{\group}{{\small \textsf{GROUP-BY}}}
\newcommand{\within}{{\small \textsf{WITHIN}}}
\newcommand{\slide}{{\small \textsf{SLIDE}}}

\newcommand{\seq}{{\small \textsf{SEQ}}}
\newcommand{\myand}{{\small \textsf{AND}}}

\newcommand{\mynext}{{\small \textsf{NEXT}}}

\newcommand{\mycount}{{\small \textsf{COUNT}}}
\newcommand{\mymin}{{\small \textsf{MIN}}}
\newcommand{\mymax}{{\small \textsf{MAX}}}
\newcommand{\mysum}{{\small \textsf{SUM}}}
\newcommand{\myavg}{{\small \textsf{AVG}}}

\usepackage[ruled]{algorithm}
\PassOptionsToPackage{noend}{algpseudocode}
\usepackage{algpseudocode}

\renewcommand{\algorithmiccomment}[1]{\bgroup\hfill//~#1\egroup}

\algnewcommand\algorithmicswitch{\textbf{switch}}
\algnewcommand\algorithmiccase{\textbf{case}}
\algnewcommand\algorithmicassert{\texttt{assert}}
\algnewcommand\Assert[1]{\State \algorithmicassert(#1)}%
\algdef{SE}[SWITCH]{Switch}{EndSwitch}[1]{\algorithmicswitch\ #1\ \algorithmicdo}{\algorithmicend\ \algorithmicswitch}%
\algdef{SE}[CASE]{Case}{EndCase}[1]{\algorithmiccase\ #1}{\algorithmicend\ \algorithmiccase}%
\algtext*{EndSwitch}%
\algtext*{EndCase}%

\usepackage[font=small,labelfont=bf,tableposition=top]{caption}

\usepackage{listings}
\renewcommand{\algorithmiccomment}[1]{/* #1 */}

\usepackage{amsthm}
\usepackage{setspace,xspace}
\graphicspath{{figures/}}

\newcommand{\nop}[1]{}

\newcommand{\rem}[1]{\marginpar{\flushleft{#1}}}
\renewcommand{\rem}[1]{} 

\renewcommand{\algorithmiccomment}[1]{/* #1 */}

\graphicspath{{figures/}}
 
\newtheorem{definition}{Definition}
\newtheorem{example}{Example}
\newtheorem{theorem}{Theorem}[section]

\usepackage{multirow}
\newcommand{\eat}[1] {}

\usepackage{fancyhdr}

\fancypagestyle{plain}{%
  \fancyhead{}
  
}

 \newlength{\hoehe}
 \newlength{\breite}
\title{\fontsize{15}{15}\selectfont Event Trend Aggregation Under Rich Event Matching Semantics\\\vspace*{1cm}
\large Technical Report\\
June 1, 2018
\vspace*{1cm}}
\author{\large Olga Poppe$^1$, Chuan Lei$^2$, Elke A. Rundensteiner$^3$, and David Maier$^4$}
\date{\Large 
\large 
$^1$Microsoft Jim Gray Systems Lab, 634 W Main St, Madison, WI 53703\\
$^2$IBM Research - Almaden, 650 Harry Rd, San Jose, CA 95120\\
$^3$Worcester Polytechnic Institute, 100 Institute Rd, Worcester, MA 01609\\
$^4$Portland State University, 1825 SW Broadway, Portland, OR 97201\\
\vspace*{3mm}
$^1$olpoppe@microsoft.com,\;$^2$chuan.lei@ibm.com,\;$^3$rundenst@wpi.edu,\;$^4$maier@pdx.edu
\vfill
}

\newcommand{\switch}{%
  \mathcode`+=\numexpr\mathcode`+ + "1000\relax 
  \mathcode`*=\numexpr\mathcode`* + "1000\relax
}

\begin{document}
\maketitle

\begin{spacing}{0.8}
{\footnotesize \noindent \textbf{Copyright} \copyright{} 2018 by
Olga Poppe. Permission to make digital or hard copies of all or
part of this work for personal use is granted without fee provided
that copies bear this notice and the full citation on the first
page. To copy otherwise, to republish, to post on servers or to
redistribute to lists, requires prior specific permission. }
\end{spacing}

\clearpage
\pagestyle{fancy}

\clearpage
\tableofcontents

\pagenumbering{arabic}
\setcounter{page}{1}

%
%

\newpage
\begin{abstract}
Streaming applications from health care analytics to algorithmic trading deploy Kleene queries to detect and aggregate event trends. Rich event matching semantics determine how to compose events into trends. The expressive power of state-of-the-art systems remains limited in that they do not support the rich variety of these semantics. Worse yet, they suffer from long delays and high memory costs because they opt to maintain aggregates at a fine granularity. 
To overcome these limitations, our Coarse-Grained Event Trend Aggregation (\app) approach supports this rich diversity of event matching semantics within one system. Better yet, \app\ incrementally maintains aggregates at the coarsest granularity possible for each of these semantics. In this way, \app\ minimizes the number of aggregates -- reducing both time and space complexity. 
Our experiments demonstrate that \app\ achieves up to four orders of magnitude speed-up and up to eight orders of magnitude memory reduction compared to state-of-the-art approaches.
\end{abstract}

\section{Introduction}
\label{sec:introduction}


Complex Event Processing (CEP) is a technology for supporting streaming applications from health care analytics to algorithmic trading. 
CEP systems continuously evaluate Kleene pattern queries against high-rate streams of primitive events to detect higher-level \textit{event trends}. 
In contrast to traditional event sequences of \textit{fixed} length~\cite{LRGGWAM11}, event trends have an \textit{arbitrary} length~\cite{PLAR17,PLRM18}.
Various event matching semantics were defined in the CEP literature to determine trend contiguity~\cite{ADGI08, WDR06, ZDI14}. 
For example, \textit{phases of a contiguously increasing heartbeat} are detected in health care analytics, while \textit{non-contiguous declining stocks} are of interest for algorithmic trading. 
Aggregation functions are applied to these trends to derive summarized insights, e.g., the \textit{maximal heartbeat} or the \textit{average price}. CEP applications must react to critical changes of these aggregates in real time.
%
%
We now describe three use cases of time-critical trend aggregation requiring diverse semantics.

$\bullet$ \textbf{\textit{Health care analytics}}.
Cardiac arrhythmia is a serious heart disease in which the heartbeat is too fast, too slow, or irregular. It can lead to life-threatening complications causing about 325K sudden cardiac deaths in US per year~\cite{bch}. Thus, the prompt detection of such abnormal heartbeat is critical to enable immediate lifesaving measures.
Query $q_1$ detects \textit{minimal} and \textit{maximal} heartbeat during passive physical activities (e.g., reading, watching TV). 
Query $q_1$ consumes a stream of heart rate measurements of intensive care patients. Each event carries a time stamp in seconds, a patient identifier, an activity identifier, and a heart rate. 
For each patient, $q_1$ detects contiguously increasing heart rate measurements during a time window of 10 minutes that slides every 30 seconds. 
No measurements may be skipped in between matched events per patient, as expressed by the \textit{contiguous} semantics.

\begin{lstlisting}[]
$q_1:\;\return\ \text{patient},\ \mymin\text{(M.rate)},\ \mymax\text{(M.rate)}$
   $\;\pattern\ \text{Measurement M+}$
   $\;\semantics\ \text{contiguous}$
   $\;\where\ \text{[patient]}\ \myand\ \text{M.rate} < \mynext\text{(M).rate}\ \myand\ \text{M.activity = passive}$
   $\;\group\ \text{patient}$
   $\;\within\ \text{10 minutes}\ \slide\ \text{30 seconds}$
\end{lstlisting}

$\bullet$ \textbf{\textit{Ridesharing service}} companies such as Uber offer peer-to-peer ridesharing with different levels of services depending on the riders' needs. With thousands of drivers and over 150 requests per minute in New York City~\cite{uber1}, real-time traffic analytics and ride management is  challenging.
%
%
Query $q_2$ computes the \textit{number of Uber pool trips} that an Uber driver can complete when some riders cancel their trips after contacting the driver during a time window of 10 minutes that slides every 30 seconds.
Each trip starts with a single \textit{Accept} event, any number of \textit{Call} and \textit{Cancel} events, followed by a single \textit{Finish} event. Each event carries a session identifier associated with the driver. All events that constitute one trip must carry the same session identifier as required by the predicate \textit{[driver]}. The \textit{skip-till-next-match} semantics allows query $q_2$ to skip irrelevant events such as in-transit, drop-off, etc. Query $q_2$ and similar queries are commonly used for Uber event stream processing and analytics, including aggregation and pattern detection, forecasting, and clustering~\cite{uber-cep}.

\begin{lstlisting}[]
$q_2:\;\return\;\text{driver},\;\mycount(*)$
   $\;\pattern\;\seq(\text{Accept},\;(\seq(\text{Call, Cancel))+,}\;\text{Finish)}$
   $\;\semantics\ \text{skip-till-next-match}$
   $\;\where\ \text{[driver]}\ \group\ \text{driver}$
   $\;\within\ \text{10 minutes}\ \slide\ \text{30 seconds}$
\end{lstlisting}

$\bullet$ \textbf{\textit{Algorithmic trading}} platforms evaluate trend aggregation queries against high-rate streams of financial transactions to identify and exploit short-term profit opportunities and avoid pitfalls. 
Stock trends of companies that belong to the same industrial sector tend to move as a group~\cite{K02}. Query $q_3$ identifies dependencies between companies within a sector to predict future trends. Each transaction carries a time stamp in seconds, a company identifier, a sector identifier, and a price. 
The query detects down-trends for a company $A$ and computes the \textit{average price} of the following trends for a company $B$ during a time window of 10 minutes that slides every 10 seconds. 
While detecting down-trends, the query ignores local price fluctuations by skipping increasing price records. Decreasing records may also be ignored to preserve opportunities for longer and thus more reliable trends~\cite{K02}. This behavior is enabled by the \textit{skip-till-any-match} semantics.

\begin{lstlisting}[]
$q_3:\;\return\ \text{sector, A.company, B.company,}\ \myavg(\text{B.price})$
   $\;\pattern\ \seq(\text{Stock A+, Stock B+})$
   $\;\semantics\ \text{skip-till-any-match}$
   $\;\where\ \text{[A.company]}\ \myand\ \text{[B.company]}\ \myand\ \text{A.price}>\mynext\text{(A).price}$
   $\;\group\ \text{sector, A.company, B.company}$
   $\;\within\ \text{10 minutes}\ \slide\ \text{10 seconds}$
\end{lstlisting}

\textbf{State-of-the-Art Approaches} to event aggregation can be divided into the following groups (Table~\ref{tab:spectrum}).

$\bullet$ \textit{\textbf{CEP approaches}} such as SASE~\cite{ZDI14}, Cayuga~\cite{DGPRSW07}, and ZStream \cite{MM09} support Kleene closure. However, Cayuga and ZStream do not consider the diverse event matching semantics.
While their languages support aggregation, they do not provide any optimization techniques to compute aggregation on top of Kleene patterns. Without  optimization, they must utilize a two-step approach that constructs all trends prior to their aggregation. This approach suffers from long delays or even fails to terminate due to the exponential time complexity of event trend construction (Section~\ref{sec:evaluation}). 

In contrast, A-Seq~\cite{QCRR14} computes aggregation of \textit{fixed-length event sequences online}, i.e., without first constructing these sequences. However, A-Seq does not support Kleene closure. Thus, it does not tackle the exponential complexity of event trends. 
\greta~\cite{PLRM18} proposes \textit{online event trend aggregation}. Since it avoids the expensive event trend construction step, it reduces the time complexity from exponential to quadratic in the number of events compared to the two-step approaches. 
However, \greta\ supports only one kind of event matching semantics, namely, skip-till-any-match.
Moreover, it maintains aggregates at the finest granularity per each matched event. We will show that its complexity is not optimal in many cases. This explains why \greta\ returns aggregation results with over an hour long delay (Section~\ref{sec:evaluation}). Such long delays are unacceptable for time-critical applications.

\vspace{-2mm}
\begin{table}[h]
\centering
\begin{tabular}{|p{0.9cm}|p{3.8cm}|p{3.2cm}|}
\hline
& \textbf{Event sequences} 
& \textbf{Event trends} \\
\hline
\textbf{Two-step}
& Flink~\cite{flink}, Esper~\cite{esper}, Oracle Stream Analytics~\cite{oracle} 
& SASE~\cite{ZDI14},~Cayuga~\cite{DGPRSW07}, ZStream~\cite{MM09} \\
\hline
\textbf{Online}
& A-Seq~\cite{QCRR14} 
& \greta~\cite{PLRM18} \\
\hline
\end{tabular}
\vspace{1mm}
\caption{State-of-the-art event aggregation approaches}
\label{tab:spectrum}
\end{table}
\vspace{-4mm}

$\bullet$ \textit{\textbf{Streaming systems}}. Industrial streaming systems such as Flink~\cite{flink}, Esper~\cite{esper}, and Oracle Stream Analytics~\cite{oracle} only support fixed-length event sequences. They do \textit{not} support Kleene closure. They construct all sequences prior to their aggregation and thus follow the two-step approach.
Streaming approaches~\cite{AW04, GHMAE07, KWF06, LMTPT05, THSW15} evaluate traditional Select-Project-Join queries, i.e., their execution paradigm is set-based. They support neither event sequences nor Kleene closure. They construct join results prior to their aggregation. Thus, they define incremental aggregation of \textit{single raw events} only.

\textbf{Challenges}.
We tackle the following open problems. 

$\bullet$ \textbf{\textit{Real-time event trend aggregation}}.
Kleene patterns match event trends of variable, statically unknown length. The number of these trends may grow exponentially in the number of events~\cite{ZDI14}. Thus, any real-time solution must aim to aggregate trends \textit{without first constructing them} and ideally even without storing all matched events. At the same time, correctness must be guaranteed, i.e., the same aggregates must be returned as by the two-step approach.

$\bullet$ \textbf{\textit{Rich event matching semantics}} were defined in the CEP literature~\cite{ADGI08, WDR06, ZDI14} to enable expressive event queries in different application scenarios. These semantics range from the most restrictive contiguous semantics (query $q_1$) to the most flexible skip-till-any-match semantics (query $q_3$). Their execution strategies differ significantly, making the seamless support of online event trend aggregation on top of these diverse semantics challenging.

$\bullet$ \textbf{\textit{Expressive predicates}} on adjacent events in a trend determine whether an event is matched depending on other events in a trend. Since a new event may be adjacent to \textit{any} previous event under the skip-till-any-match semantics, \textit{all} matched events must be kept. The need to store all matched events contradicts the online aggregation requirement that aims to incrementally update aggregates upon event arrival and discard these events immediately thereafter. 

\begin{figure}[t]
\centering
\includegraphics[width=0.5\columnwidth]{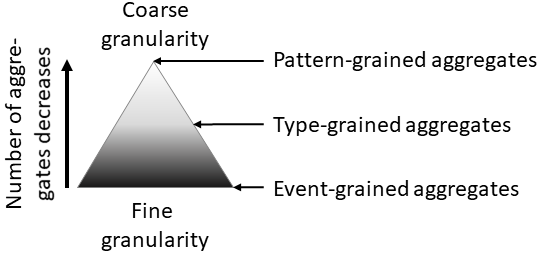}  
\vspace{-2mm}
\caption{Event trend aggregation at multiple granularities}
\label{fig:idea}
\end{figure} 

\textbf{Our Proposed COGRA Approach} is the first to define online event trend aggregation under rich event matching semantics at multiple granularities. 
\app\ pushes aggregation inside Kleene closure computation. Thus, it avoids the event trend construction step with exponential complexity. 
Better yet, depending on the event matching semantics, \app\ adaptively selects the \textit{coarsest possible granularity} at which it incrementally computes trend aggregation. These granularities range from coarse (per pattern), to medium (per event type), to fine (per matched event) as per Figure~\ref{fig:idea}. \app\ minimizes the number of aggregates to be maintained and discards all events as soon as they have been used to update the aggregates. Thus,
our approach represents a win-win solution that reduces both time and space complexity of trend aggregation compared to state-of-the-art~\cite{flink, PLRM18, QCRR14, ZDI14}.

\textbf{Contributions}. 
The key innovations of \app\ include:

1)~We define the problem of real-time event trend aggregation under rich event matching semantics. Based on these semantics and other query features, we determine the coarsest granularity at which trend aggregates are maintained.

2)~For each granularity, we propose efficient data structures and algorithms to incrementally compute event trend aggregation. We prove the correctness of these algorithms.

3)~Our \app\ strategies are shown to reduce both time and space complexity compared to state-of-the-art approaches. We prove that \app\ achieves optimal time complexity. 

4)~Our experiments using synthetic and real data sets~\cite{stockStream,RS12} demonstrate that our \app\ approach achieves up to four orders of magnitude speed-up and uses up to eight orders of magnitude less memory compared to Flink~\cite{flink}, SASE~\cite{ZDI14}, \greta~\cite{PLRM18}, and A-Seq~\cite{QCRR14} (Table~\ref{tab:spectrum}).

\textbf{Outline}. 
Section~\ref{sec:model} describes our data and query model.
Section~\ref{sec:overview} provides an overview of the \app\ framework.
Depending on event matching semantics, Sections~\ref{sec:type-grained-aggregator}--\ref{sec:pattern-grained-aggregator} define incremental trend aggregation at different granularities.
We consider the remaining query clauses in Section~\ref{sec:other}. 
We discuss how to extend our event query language in Section~\ref{sec:discussion}. 
Section~\ref{sec:evaluation} describes our experimental study. 
Section~\ref{sec:related_work} discusses related work, while
Section~\ref{sec:conclusions} concludes the paper.

\section{Data and Query Model}
\label{sec:model}

\subsection{Basic Notions and Assumptions}
\label{sec:basic}

\textbf{\textit{Time}} is represented by a linearly ordered set of time points $(\mathbb{T},\leq)$, where $\mathbb{T} \subseteq \mathbb{Q^+}$ (the non-negative rational numbers). 
An \textbf{\textit{event}} is a message indicating that something of interest to the application happened in the real world. An event $e$ has a \textbf{\textit{time stamp}} $e.time \in \mathbb{T}$ assigned by the event source. 
An event $e$ belongs to a particular \textbf{\textit{event type}} $E$, denoted \textit{e.type=E} and described by a \textit{schema} that specifies the set of \textit{event attributes} and the domains of their values.
Events are sent by event producers (e.g., sensors) on an \textbf{\textit{event stream}} $I$. An event consumer (e.g., health care system) continuously monitors the stream with \textit{event queries}. We borrow the query language and semantics from SASE~\cite{ADGI08, WDR06, ZDI14}. 
Our example queries $q_1$--$q_3$ in Section~\ref{sec:introduction} are expressed using this syntax.

\begin{definition}[\textbf{\textit{Pattern}}]
A \textit{\textbf{pattern}} has the form
$E$, $P+$, or \seq$(P_1,P_2)$,
where 
$E$ is an event type,
$P,P_1,P_2$ are patterns, 
$+$ is a Kleene plus operator, and 
\seq\ is an event sequence operator.
$P$ is a \textit{\textbf{sub-pattern}} of $P+$,
while $P_1$ and $P_2$ are \textit{\textbf{sub-patterns}} of $\seq(P_1,P_2)$.
If an pattern contains a Kleene plus operator, it is called a \textbf{\textit{Kleene pattern}}. It is an \textbf{\textit{event sequence pattern}} otherwise.
The length of a pattern is the number of event types in it. 
\label{def:pattern}
\end{definition}

We focus on Kleene patterns that allow us to specify arbitrarily long event pattern matches, called \textbf{\textit{event trends}} (Definitions~\ref{def:event-trend-under-STAM}--\ref{def:event-trend-under-CONT}).
To simplify our discussion, we consider Kleene patterns that do not contain negation, Kleene star, optional sub-patterns, conjunction, nor disjunction. Also, an event type may appear at most once in a pattern. We sketch straightforward extensions of our approach to relax these assumptions in Section~\ref{sec:discussion}.

\subsection{Event Matching Semantics}
\label{sec:semantics}

Event matching semantics~\cite{ADGI08, WDR06, ZDI14} constrain event contiguity in a trend to express queries for diverse streaming applications (Section~\ref{sec:introduction}). These semantics differentiate between \textbf{\textit{relevant events}}, i.e., events that can extend an existing partial trend (Definition~\ref{def:finished-partial-trend}), and \textbf{\textit{irrelevant events}} that cannot. Relevant events either must extend existing trends or can be skipped to preserve opportunities for alternative trends. Irrelevant events either invalidate partial trends or can be skipped. 

\begin{figure}[t]
\centering
\includegraphics[width=0.6\columnwidth]{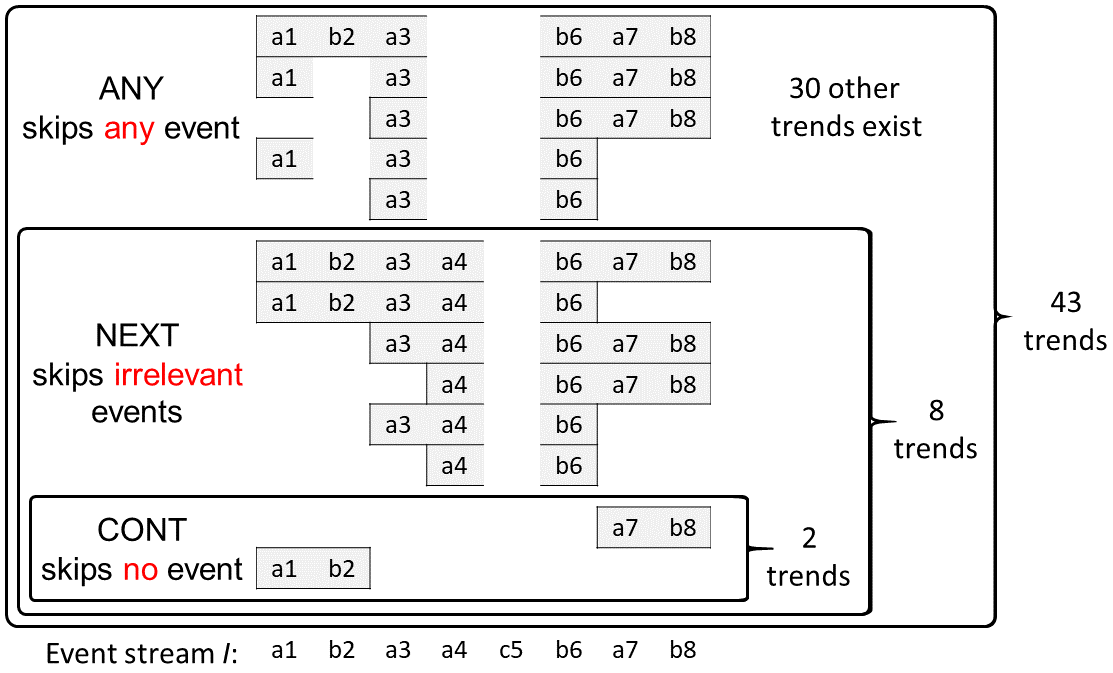}  
\caption{Event trends matched by the Kleene pattern \textit{P=(\seq(A+,B))+} under various event matching semantics}
\label{fig:ess}
\end{figure} 

\begin{example}
In Figure~\ref{fig:ess}, the pattern \textit{P = (\seq(A+,B))+} is evaluated under various event matching semantics against the stream $I$ (depicted at the bottom of the figure). In the stream, letters denote types, while numbers represent time stamps, e.g., $a1$ is an event of type $A$ with time stamp 1. Matched trends are depicted above the stream. They range from the shortest contiguous trend $(a1,b2)$%
\footnote[2]{Constraints on minimal trend length exclude too short and thus not meaningful event trends. They are considered in Section~\ref{sec:discussion}.}
to the longest non-contiguous trend $(a1,b2,a3,a4,b6,a7,b8)$.
\end{example}

\textbf{Skip-Till-Any-Match Semantics} (\anysem) is the most flexible semantics that detects \textit{all possible trends} as follows. 
For each event $e$ and each partial trend $tr$ that can be extended by $e$, two possibilities are considered: 
(1)~$e$ is appended to the trend $tr$ to form a longer trend $tr'=(tr,e)$, and 
(2)~$e$ is skipped and the trend $tr$ remains unchanged to preserve opportunities for alternative longer trends. 
If an event $e$ can extend all trends, then $e$ doubles the number of trends. Thus, the number of trends grows exponentially in the number of events in the worst case.
Skip-till-any-match skips irrelevant events.
Query $q_3$ in Section~\ref{sec:introduction} is evaluated under this semantics.

\begin{example}
In Figure~\ref{fig:ess}, when $a7$ arrives, the trend $(a3,$ $b6)$ is extended to $(a3,b6,a7)$ and the original trend $(a3,b6)$ is also kept. Based on only eight events in the stream, 43 trends are detected. Only some of them are shown for compactness. Irrelevant events are ignored, e.g., $c5$. 
\end{example}

\begin{definition}[\textbf{\textit{Event Trend Under Skip-Till-Any-Match}}]
An event type $E$ matches an event $e \in I$ of type $E$ under skip-till-any-match, 
denoted $e \in \mathit{trends}_\mathit{any}(E,I)$.

If $P_1$ and $P_2$ are patterns,
$(e_1,\dots,e_m) \in \mathit{trends}_\mathit{any}(P_1,I)$,
$(e_{m+1},\dots,e_k) \in \mathit{trends}_\mathit{any}(P_2,I)$, and
$e_m.time < e_{m+1}.$ $time$, then
the event sequence pattern \seq $(P_1,P_2)$ matches the trend
$s=(e_1,\dots,e_k)$ under skip-till-any-match, 
denoted $s \in \mathit{trends}_\mathit{any}(\seq(P_1,P_2),I)$.

If $P$ is a pattern, 
$\forall s_l \in \{s_1, \dots, s_k\}$.
$s_l \in \mathit{trends}_\mathit{any}(P,I)$ and 
$s_l.end.time < s_{l+1}.start.time$, then
the Kleene plus pattern $P+$ matches the trend
$tr=(s_1 : \dots : s_k)$ under skip-till-any-match,
denoted $tr \in \mathit{trends}_\mathit{any}(P+,I)$, 
where ``:" is concatenation.
Start, mid, and end events of a trend are defined in Table~\ref{tab:start-mid-end}.
\label{def:event-trend-under-STAM}
\end{definition}

\vspace*{-3mm}
\begin{table}[h]
\centering
\begin{tabular}{|l||l|p{3.2cm}|l|}
\hline
Trend $tr$
& $tr.start$
& $tr.mid$
& $tr.end$
\\
\hline\hline
$e$
& $e$
& $\emptyset$
& $e$
\\
\hline
$(e_1,\dots,e_k)$
& $e_1$
& $\{e_2,\dots,e_{k-1}\}$
& $e_k$
\\
\hline
$(s_1 : \dots : s_k)$ 
& $s_1.start$
& $\{ e \mid e \text{ is in } tr, e \neq s_1.$ $start, e \neq s_k.end \}$
& $s_k.end$
\\
\hline
\end{tabular}
\vspace*{1mm}
\caption{\textit{Start}, \textit{mid}, and \textit{end} events of a trend}
\label{tab:start-mid-end}
\end{table}
\vspace*{-4mm}

\textbf{Skip-Till-Next-Match Semantics} (\nextsem) is more restrictive than \anysem\ because \nextsem\ requires that all relevant events are matched. It allows skipping \textit{irrelevant} events however. Query $q_2$ in Section~\ref{sec:introduction} is evaluated under this semantics. 

\begin{example}
In Figure~\ref{fig:ess}, the trend $(a3,b6)$  does not conform to this semantics since it skipped over the relevant event $a4$. In contrast, the trend $(a3,a4,b6)$ is valid since it skips no relevant events in between matched events.
\end{example}

\begin{definition}[\textbf{\textit{Event Trend Under Skip-Till-Next-Match}}]
If $tr \in \mathit{trend}_\mathit{any}(P,I)$ and
$\nexists tr' \in \mathit{trend}_\mathit{any}(P,$ $I)$ with
$tr.start = tr'start$, 
$tr.end = tr'.end$, and 
$tr.mid \subseteq tr'.mid$, then
the pattern $P$ matches the trend $tr$ under skip-till-next-match, 
denoted $tr \in \mathit{trends}_\mathit{next}(P,I)$.
\label{def:event-trend-under-STNM}
\end{definition}

\textbf{Contiguous Semantics} (\contsem) is the most restrictive semantics since it \textit{does not skip events}. Query $q_1$ in Section~\ref{sec:introduction} detects contiguous trends.

\begin{example}
In Figure~\ref{fig:ess}, $(a1,b2)$ and $(a7,b8)$ are the only contiguous trends. Since $c5$ cannot be ignored, $a1$--$a4$ cannot form contiguous trends with later events.
\end{example}

\begin{definition}[\textbf{\textit{Event Trend Under Contiguous Semantics}}]
If $tr \in \mathit{trends}_\mathit{next}(P,I)$ and
$\nexists e \in I$ such that 
$tr.start.time < e.time < tr.end.time$ and 
$e$ is not part of the trend $tr$, then
the pattern $P$ matches the trend $tr$ under the contiguous semantics, 
denoted $tr \in \mathit{trends}_\mathit{cont}(P,I)$.
\label{def:event-trend-under-CONT}
\end{definition}

Figure~\ref{fig:ess} illustrates the containment relations among the sets of trends matched by the pattern $P$ under various semantics.

\begin{definition}[\textbf{\textit{Finished vs Partial Event Trend}}]
Given a pattern $P$ and its sub-pattern $P'$ evaluated under an event matching semantics $S$, $tr \in \mathit{trends}_S(P,I)$ is a \textbf{\textit{finished trend}} of $P$, while $tr' \in \mathit{trends}_S(P',I)$ is a \textbf{\textit{partial trend}} of $P$.
\label{def:finished-partial-trend}
\end{definition}

\subsection{Event Trend Aggregation Query}
\label{sec:query}

An event trend aggregation query constrains the trends that are detected under various event matching semantics by predicates, grouping, and windows as follows.

\begin{definition}[\textbf{\textit{Event Trend Aggregation Query}}]
An \textit{\textbf{event trend aggregation query}} $q$ consists of six clauses:

$\bullet$ Aggregation result specification (\return\ clause),

$\bullet$ Kleene pattern $P$ (\pattern\ clause),

$\bullet$ Event matching semantics $S$ (\semantics\ clause),

$\bullet$ Predicates $\theta$ (optional \where\ clause),

$\bullet$ Grouping $G$ (optional \group\ clause), and

$\bullet$ Window $w$ (\within\ and \slide\ clause).

If a trend $tr$ is matched by the pattern $P$ 
under the semantics $S$,
all events in $tr$ satisfy the predicates $\theta$, 
carry the same values of the grouping attributes $G$, and 
are within one window $w$, 
then the trend $tr$ is matched by the query $q$.
\label{def:query}
\end{definition}

Within each window of query $q$, matched trends are grouped by the values of grouping attributes $G$. Aggregates are computed per group. We focus on distributive (such as \mycount, \mymin, \mymax, \mysum) and algebraic aggregation functions (such as \myavg) since they can be computed incrementally~\cite{Gray97}. 

$\bullet$ \textbf{\textit{Count}}.
\mycount$\mathit{(*)}$ returns the number of trends per group.
Let \textit{tr.}\mycount \textit{(E)} be the number of events of type $E$ in a trend $tr$. \mycount\textit{(E)} corresponds to the sum of \textit{tr.}\mycount$(E)$ of all trends $tr$ per group.

$\bullet$ \textbf{\textit{Minimum and Maximum}}.
Let \textit{tr.}\mymin\textit{(E.attr)} be the minimal value of an attribute \textit{attr} of events of type $E$ in a trend $tr$.
\mymin\textit{(E.attr)} returns the minimal value of \textit{tr.}\mymin$(E.attr)$ of all trends $tr$ per group.
\mymax\textit{(E.attr)} is defined analogously to \mymin$(E.\mathit{attr})$.

$\bullet$ \textbf{\textit{Summation and Average}}.
Let \textit{tr.}\mysum\textit{(E.attr)} be the sum of values of an attribute \textit{attr} of events of type $E$ in a trend $tr$.
\mysum\textit{(E.attr)} corresponds to the sum of \textit{tr.}\mysum\textit{(E.attr)} of all trends $tr$ per group.
Lastly, \myavg\textit{(E.attr)} = \mysum\textit{(E.attr)} / \mycount\textit{(E)} per group.

In this paper, we illustrate trend count computation, i.e., \mycount (*). The same principles apply to other aggregation functions (Section~\ref{sec:discussion}).
\section{Cogra Approach Overview}
\label{sec:overview}

To support time-critical streaming applications (Section~\ref{sec:introduction}), we solve the following \textbf{\textit{Event Trend Aggregation Problem}}. 
Given an event trend aggregation query $q$ (Definition~\ref{def:query}) evaluated under an event matching semantics (Section~\ref{sec:semantics}) over an event stream $I$, our goal is to compute the aggregation results of $q$ with \textit{minimal latency}.

\begin{figure}[h]
\centering
\includegraphics[width=0.6\columnwidth]{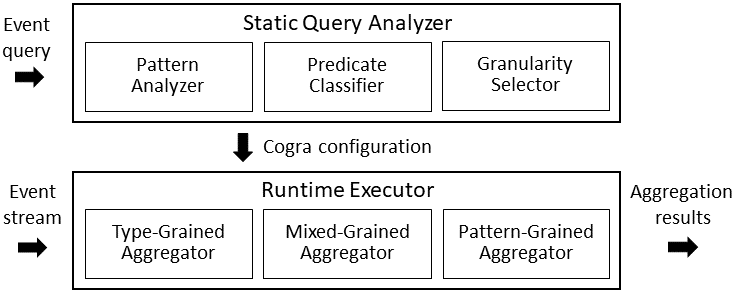}  
\caption{\app\ framework}
\label{fig:system}
\end{figure} 

Figure~\ref{fig:system} illustrates our \app\ framework.
In order to select the granularity at which the aggregates are maintained for a query $q$ (Section~\ref{sec:granularity-selector}), 
the \textbf{\textit{Static Query Analyzer}}
analyzes the pattern of $q$ (Section~\ref{sec:pattern-analyzer}) and
classifies the predicates of $q$ (Section~\ref{sec:predicate-classifier}). 
The results of this query analysis are encoded into the \app\ configuration to guide our \textbf{\textit{Runtime Executor}} (Sections~\ref{sec:pattern-grained-aggregator}--\ref{sec:mixed-grained-aggregator}).

\subsection{Pattern Analyzer} 
\label{sec:pattern-analyzer}

To facilitate the analysis of a pattern $P$, we translate $P$ into its Finite State Automaton (FSA)-based representation~\cite{ADGI08, DGPRSW07, MM09, WDR06, ZDI14}. We informally describe this translation here, while the algorithm can be found in the extended version~\cite{PLRM18}.

\textbf{\textit{States}} are labeled by event types in $P$. The first state is the \textbf{\textit{start type}} \textit{start(P)} and the final state is the \textbf{\textit{end type}} \textit{end(P)}. All other states are labeled by \textbf{\textit{middle types}} \textit{mid(P)}. According to our assumption in Section~\ref{sec:basic}, a type may occur at most once in $P$. Thus, state labels are distinct. There is exactly one start type, exactly one end type, and any number of middle types in $P$~\cite{PLRM18}. 
In Figure~\ref{fig:template}, \textit{start(P) = A}, \textit{end(P) = B}, and $\mathit{mid(P)} = \emptyset$, meaning that a trend matched by $P$ always starts with an $a$ and ends with a $b$.

\textbf{\textit{Transitions}} are labeled by operators in $P$. They connect types of events that are adjacent in a trend matched by $P$ (Definition~\ref{def:adjacent_events}). If a transition connects a type $E'$ with a type $E$, then $E'$ is called a \textit{\textbf{predecessor type}} of $E$, denoted $E' \in P.\mathit{predTypes}(E)$. 
In Figure~\ref{fig:template}, $P.\mathit{predTypes}(A) = \{A,B\}$ and $P.\mathit{predTypes}(B) = \{A\}$, meaning that an $a$ may be preceded by previously matched $a$'s and $b$', while a $b$ is preceded by previously matched $a$'s.

\begin{figure}[t]
\centering
\includegraphics[width=0.2\columnwidth]{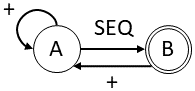}  
\caption{FSA representation of the pattern \textit{P=(\seq(A+,B))+}}
\label{fig:template}
\end{figure}

\begin{definition}[\textbf{\textit{Adjacent Events, Predecessor Event}}]
Let $e_p,e \in I$ be events such that $e_p$ is in a partial trend matched by a query $q$ and $e$ is new. 
The events $e_p$ and $e$ are \textbf{\textit{adjacent under the skip-till-any-match semantics}} in a window $w$ if the following conditions hold:
(1)~$e_p.type \in P.predTypes(e.type)$,
(2)~$e_p.time < e.time$,
(3)~$e_p$ and $e$ satisfy the predicates $\theta$,
(4)~$e_p$ and $e$ have the same values of grouping attributes $G$, and
(5)~$e_p$ and $e$ belong to the window $w$.

The events $e_p$ and $e$ are \textbf{\textit{adjacent under the skip-till-next-match semantics}} in a window $w$ if conditions 1--5 above hold and 
$\nexists e' \in I$ such that $e'.time < e.time$ and $e_p$ and $e'$ are adjacent in $w$ under skip-till-any-match.

The events $e_p$ and $e$ are \textbf{\textit{adjacent under the contiguous semantics}} in a window $w$ if conditions 1--5 above hold and 
$\nexists e' \in I$ such that $e_p.time < e'.time < e.time$.

If $e_p$ and $e$ are adjacent in a window $w$, then 
$e_p$ is called a \textbf{\textit{predecessor event}} of $e$ in $w$.
\label{def:adjacent_events}
\end{definition}

\subsection{Predicate Classifier} 
\label{sec:predicate-classifier}

We distinguish between predicates on single events and predicates on adjacent events since they determine the granularity at which aggregates are maintained (Section~\ref{sec:granularity-selector}).

\textbf{\textit{Predicates on single events}} either filter or partition the stream. 
For example, query $q_1$ in Section~\ref{sec:introduction} uses the predicate \textit{(M.activity = passive)} to selects those measurements that were taken during passive activities and the predicate \textit{[patient]} to partition the stream by patient identifier.

\textbf{\textit{Predicates on adjacent events}} restrict the adjacency relation between events in a trend. 
For example, the predicate $(M.\mathit{rate}<\mynext(M).\mathit{rate})$ of $q_1$ requires heartbeat measurements to increase from one event to the next in a trend. 


\subsection{Granularity Selector} 
\label{sec:granularity-selector}

The number of event trends matched by a pattern $P$ is determined by the presence of Kleene plus in $P$ and the semantics under which $P$ is evaluated~\cite{QCRR14, ZDI14} (Table~\ref{tab:complexity}). The number of trends matched by $P$ ranges from linear to exponential in the number of matched events.

State-of-the-art two-step approaches~\cite{oracle, DGPRSW07, esper, flink, ZDI14} first construct \textit{all event trends} (Figure~\ref{fig:ess}) and then compute their aggregation. These approaches do not offer a feasible real-time solution, since they suffer from exponential overhead of event trend construction in the worst case (Table~\ref{tab:complexity}). 

\begin{table}[t]
\centering
\begin{tabular}{|c||c|c|}
\hline
\textbf{Event matching}
& \textbf{Event sequence} 
& \textbf{Kleene} 
\\
\textbf{semantics}
& \textbf{pattern}
& \textbf{pattern} 
\\
\hline\hline
\anysem
& Polynomial
& Exponential \\
\hline
\nextsem, \contsem
& Linear
& Polynomial \\
\hline
\end{tabular}
\vspace*{1mm}
\caption{Number of trends in the number of events}
\label{tab:complexity}
\end{table}

To overcome these limitations, we propose to \textit{omit the trend construction step} and \textit{\textbf{incrementally}} compute trend aggregation. All events are discarded once they have been used to update aggregates.
Better yet, depending on the semantics of a query $q$, our granularity selector chooses the \textit{\textbf{coarsest granularity}} at which trend aggregates are maintained by $q$ such that both \textit{correctness} and \textit{optimal time complexity} of trend aggregation are guaranteed for $q$. 
More precisely, our granularity selector decides whether to maintain trend aggregates for $q$ at pattern, type, or mixed granularity as follows (Table~\ref{tab:granularity-selection}).

\begin{table}[h]
\centering
\begin{tabular}{|c||c|c|}
\hline
\textbf{Event matching}
& \multicolumn{2}{c|}{\textbf{Predicates on adjacent events}} 
\\
\textbf{semantics}
& \multicolumn{1}{c|}{\ \ \ \ \ \textbf{without}\ \ \ \ \ }
& \multicolumn{1}{c|}{\textbf{with}} \\
\hline\hline
\anysem
& \multicolumn{1}{c|}{Type}
& \multicolumn{1}{c|}{Mixed} \\
\hline
\nextsem,\;\contsem
& \multicolumn{2}{c|}{Pattern} \\
\hline
\end{tabular}
\vspace*{1mm}
\caption{Granularity selection}
\label{tab:granularity-selection}
\end{table}

\textbf{\textit{Type-grained aggregator}}.
If the query $q$ is evaluated under the skip-till-any-match semantics and has no predicates on adjacent events, our executor maintains an aggregate per each event type in the pattern
(Section~\ref{sec:type-grained-aggregator}).

\textbf{\textit{Mixed-grained aggregator}}.
If $q$ is evaluated under the skip-till-any-match semantics and has predicates on adjacent events $\theta$, our executor maintains the aggregates at mixed granularities, namely, either per event $e$ if $e$ is required to evaluate the predicates $\theta$ or per event type \textit{e.type} otherwise 
(Section~\ref{sec:mixed-grained-aggregator}).

\textbf{\textit{Pattern-grained aggregator}}.
If the query $q$ is evaluated under the contiguous or skip-till-next-match semantics, our executor adopts the pattern-grained aggregation strategy. Namely, only the final aggregate of $q$ and the intermediate aggregate of the last matched event are kept
(Section~\ref{sec:pattern-grained-aggregator}).

\section{Type-Grained Aggregator}
\label{sec:type-grained-aggregator}

To overcome \textit{exponential time overhead} of trend construction under skip-till-any-match  (Table~\ref{tab:complexity}), we now propose to \textit{incrementally} compute trend aggregation at the \textit{event type granularity}.
Coarse-grained trend aggregation under skip-till-any-match is complicated by the flexibility of this semantics, especially, by its ability to skip \textit{any} event in between adjacent events in a trend (Definition~\ref{def:event-trend-under-STAM}). Thus, \textit{any} previously matched event $e_p$ may be a predecessor event of a new event $e$ in a trend (Definition~\ref{def:adjacent_events}). In particular, each previously matched event $e_p$ must be kept in order to evaluate predicates on adjacent events while deciding whether $e_p$ and a new event $e$ are adjacent. This requirement contradicts our goal of incrementally computing trend aggregation and immediately discarding all events. 

However, if the query has \textit{no predicates on adjacent events}, we now postulate that incremental trend aggregation is possible at the type granularity as follows. 
Windows, predicates on single events, and grouping partition the stream into sub-streams (Section~\ref{sec:other}). Let $e$ be an event of type $E$ within a sub-stream $S$. When $e$ arrives, \textit{all} previously matched events of predecessor types of $E$ are adjacent to $e$ within the sub-stream $S$ (Definition~\ref{def:adjacent_events}). Thus, a count can be assigned to each type in the pattern $P$. The event $e$ updates the count of $E$ and is discarded thereafter. The final count corresponds to the count of the end type of $P$. 

\vspace{-2mm}
\begin{table}[h]
\centering
\begin{tabular}{|c||c|c|c|}
\hline
Event $e$
& $e.count$
& $A.count$
& $B.count$ \\
\hline
\hline
a1 & 1 & 1 & \\ \hline
b2 & 1 & & 1 \\ \hline
a3 & 3 & 4 & \\ \hline 
a4 & 6 & 10 & \\ \hline 
b6 & 10 & & 11 \\ \hline 
a7 & 22 & 32 & \\ \hline 
b8 & 32 &  & 43 \\ \hline 
\end{tabular}
\vspace{2mm}
\caption{Type-grained trend count}
\label{tab:type}
\end{table}
\vspace{-6mm}

\begin{example}
Continuing our running example in Figure~\ref{fig:ess}, the type-grained trend count computation is shown in Table~\ref{tab:type}. 
For example, when an event $a7$ is matched, the intermediate count $a7.count$ captures the number of partial trends that end at $a7$. According to our predecessor relationship analysis in Section~\ref{sec:pattern-analyzer}, all previously matched $a$'s and $b$'s are adjacent to $a7$. Thus, $a7.count$ is set to the sum of the counts of all previously matched $a$'s and $b$'s to accumulate the number of trends extended by $a7$. We further increment $a7.count$ by one since $a7$ is of a start type of the pattern $P$ and thus begins one new trend (Section~\ref{sec:pattern-analyzer}). 

$A.count$ accumulates the number of partial trends that end at an $a$. Thus, each $a$ increments $A.count$ by $a.count$. 

$B.count$ is computed analogously to $A.count$. Since a $b$ is of an end type of the pattern $P$, a $b$ finishes all trends it extends. Thus, $B.count$ corresponds to the final count. 43 trends are detected (Figure~\ref{fig:ess}). 
\[ 
\begin{array}{ll}
a7.count &= A.count+B.count+1 = 10+11+1 = 22. \\
A.count &= A.count+a7.count = 10+22 = 32. \\
b8.count &= A.count = 32. \\
B.count &= B.count+b8.count = 11+32 = 43.
\end{array}
\]

Only updates are shown in Table~\ref{tab:type} for readability. Only the most recent values of the type-grained aggregates $A.count$ and $B.count$ are stored at a time.
\label{ex:type}
\end{example}

\begin{theorem}[\textbf{\textit{Type-Grained Trend Count}}]
Let $q$ be a query that is evaluated under skip-till-any-match and has no predicates on adjacent events, $P$ be its pattern, and $e \in I$ be an event of type $E$.
Then the intermediate count $e.count$ corresponds to the number of (partial) trends that end at the event $e$. 
The type-grained trend count $E.count$ captures the number of (partial) trends that end at an event of type $E$.
The final count is the number of finished trends matched by $q$.
\[
\begin{array}{ll}
e.\mathit{count} &= \sum_{E' \in P.\mathit{predTypes}(E)} E'.\mathit{count}.\\
E.\mathit{count} &= \sum_\mathit{e.type=E} e.\mathit{count}.\\
\mathit{final}\_\mathit{count} &= end(P).\mathit{count}.
\end{array}
\]
If $E=start(P)$, $e.count$ is incremented by one.
\label{theo:correctness-type}
\end{theorem}


\begin{proof}
By Definition~\ref{def:adjacent_events}, event adjacency is determined by windows, grouping, predicates, event types, and event time stamps. Windows, grouping, and predicates on single events partition the event stream such that trend aggregation is computed for each partition separately as defined in Section~\ref{sec:other}. We now prove the statement by induction on the number of matched events $n$ for the remaining query constraints.

\textit{Induction Basis}: $n=1$. If only one event $e$ of type $E$ is matched, then there are no trends to extend and $e$ starts a new trend. Thus, $e.count=1$ and $E.count=1$. If $E$ is the end type of $P$, $\mathit{final}\_\mathit{count}=1$. Otherwise, $\mathit{final}\_\mathit{count}=0$.

\textit{Induction Assumption}: This statement holds for $n$ events.

\textit{Induction Step}: $n \rightarrow n+1$. 
According to the induction assumption, for a type $E'$ in $P$, $E'.count$ corresponds to the number of (partial) trends that end at an event of type $E'$ after $n$ events have been processed. Let $e$ be a new event of type $E$. By Definition~\ref{def:event-trend-under-STAM} without predicates on adjacent events $e$ is adjacent to all previously matched events of type $E' \in P.\mathit{predTypes}(E)$, i.e., $e$ continues all these trends. To accumulate the number of trends extended by $e$, $e.count$ corresponds to the sum of counts of all predecessor types of $E$, i.e., $E'.count$. In addition, if $E$ is a start type of $P$, $e$ begins a new trend and $e.count$ is incremented by 1. 
$E.count$ is increased by $e.count$ to accumulate the number of trends extended by $e$.
Lastly, the query $q$ counts the number of finished trends only. By Definition~\ref{def:finished-partial-trend} and Section~\ref{sec:pattern-analyzer}, only events of end type of $P$ may finish trends. Thus, the count of the end type of $P$ captures the number of finished trends.
\end{proof}

\begin{algorithm}[t]
\caption{Type-grained trend count algorithm}
\label{algo:type}
\begin{algorithmic}[1]
\Require Query $q$ with pattern $P$, event stream $I$
\Ensure Count of event trends matched by $q$ in $I$
\State $H \leftarrow \text{empty hash table}$
\ForAll {event type $E$ in $P$} $H.put(E,0)$
\EndFor
\ForAll {$e \in I$ of type $E$} 
	\State $e.\mathit{count} \leftarrow (E = \mathit{start}(P))\ ?\ 1\ :\ 0$
	\ForAll {$E' \in P.\mathit{predTypes}(E)$}  
		\State $\switch e.\mathit{count} += H.get(E')$ 	 		
  	\EndFor   			
  	\State $E.count \leftarrow H.get(E) + e.count$
    \State $H.put(E, E.count)$  		
\EndFor
\State\Return $H.get(end(P))$
\end{algorithmic}
\end{algorithm}

\textbf{Type-Grained Trend Count Algorithm} maintains a hash table $H$ that maps each type $E$ in the pattern $P$ to the count for $E$. Initially, all counts are set to 0 (Lines~1--2 in Algorithm~\ref{algo:type}).
For each event $e$ of type $E$, $e.count=1$ if $E$ is a start type of $P$. Otherwise, $e.count=0$ (Lines~3--4). For each predecessor type $E'$ of $E$, $e.count$ is incremented by $E'.count$ (Lines~5--6). $E.count$ is incremented by $e.count$ (Lines~7--8). The count of the end type of $P$ is returned (Line~9).

\begin{theorem}[\textbf{\textit{Complexity}}]
Let $q$ be a query evaluated under skip-till-any-match, $P$ be its pattern of length $l$, and $n$ be the number of events per window of $q$. Algorithm~\ref{algo:type} has linear time $O(nl)$ and space $\Theta(l)$ complexity.
\label{theo:comlexity-type}
\end{theorem}

\begin{proof}
For each matched event of type $E$, the type-grained trend counts of all predecessor types of $E$ are accessed. Thus, the time complexity of Algorithm~\ref{algo:type} is linear: $O(nl)$.
Its space complexity is determined by the number of counts. Since one count is stored per type, the space costs are linear: $\Theta(l)$.
\end{proof}

\begin{theorem}[\textbf{\textit{Time Optimality}}]
The linear time complexity $O(nl)$ of Algorithm~\ref{algo:type} is optimal.
\label{theo:optimality-type}
\end{theorem}

\begin{proof}
Any trend aggregation algorithm must process $n$ events to ensure correctness of the final count.
By Definition~\ref{def:adjacent_events}, event adjacency is determined by event types. In the worst case, all types in $P$ are predecessor types of a type $E$.
%
%
There are $\Theta(l)$ types in $P$. One count is maintained per type. Each matched event of type $E$ triggers an update of $E.count$ that accesses $O(l)$ type-grained trend counts. Thus, the linear time complexity $O(nl)$ is optimal.
\end{proof}

\section{Mixed-Grained Aggregator}
\label{sec:mixed-grained-aggregator}

We now extend our coarse-grained trend aggregation techniques to the \textit{most general class of queries under skip-till-any-match with predicates on adjacent events} $\theta$, for example, query $q_3$ in Section~\ref{sec:introduction}. 
To this end, we propose to maintain aggregates at mixed granularities. 
Namely, we classify the event types in a pattern $P$ into two disjoint sets $\mathcal{T}_e$ and $\mathcal{T}_t$.
Events of types $\mathcal{T}_e$ must be stored to evaluate the predicates $\theta$ as new events arrive. Thus, an \textit{event-grained trend aggregate} is computed for each event of type in $\mathcal{T}_e$.
In contrast, events of types $\mathcal{T}_t$ do not have to be kept. Thus, a \textit{type-grained trend aggregate} is maintained for each type in $\mathcal{T}_t$. 
Only in the extreme case when expressive predicates restrict all events (i.e., $\mathcal{T}_t = \emptyset$), fine-grained trend aggregation is required~\cite{PLRM18}.

\vspace{-2mm}
\begin{table}[h]
\centering
\begin{tabular}{|c||c|c|c|}
\hline
Event $e$
& $e.count$
& $A.count$
& $\mathit{final}\_count$ \\
\hline
\hline
a1 & 1 & 1 & \\ \hline
b2 & 1 & & 1 \\ \hline
a3 & 3 & 4 & \\ \hline 
a4 & 6 & 10 & \\ \hline 
b6 & 10 & & 11 \\ \hline 
a7 & 12 & 22 & \\ \hline 
b8 & 22 &  & 33 \\ \hline 
\end{tabular}
\vspace{2mm}
\caption{Mixed-grained trend count: Type-grained trend count for $A$ and event-grained trend counts for $b$'s}
\label{tab:mixed}
\end{table}
\vspace{-6mm}

\begin{example}
Continuing our example in Figure~\ref{fig:ess}, assume that the predicates $\theta$ restrict the adjacency relations between $b$'s and $a$'s. When an $a$ arrives, we have to compare it to \textit{each} previously matched $b$ to select those $b$'s that satisfy the predicates $\theta$. Thus, \textit{event-grained trend counts} must be maintained for $b$'s (Table~\ref{tab:mixed}).
In contrast, all $a$'s are adjacent to following $b$'s. Thus, a \textit{type-grained trend count} can be maintained for type $A$. Assuming that $a7$ is adjacent to $b2$ but not to $b6$, $a7.count$ is computed based on the mixed-grained trend counts as follows:
\[
\begin{array}{ll}
a7.count 
&= A.count + \sum_{(b,a7) \text{ satisfy } \theta} b.count + 1 \\
&= A.count + b2.count + 1 = 10+1+1 = 12.
\end{array}
\]

All matched $b$'s, their counts, and the most recent values of $A.count$ and $\mathit{final}\_count$ are stored.
%
%
\label{ex:mixed}
\end{example}

\begin{theorem}[\textbf{\textit{Mixed-Grained Trend Count}}]
Let query $q$ be evaluated under skip-till-any-match,
$\theta$ be its predicates on adjacent events,
$P$ be its pattern,
\textit{attr} and $\mathit{attr}_x$ be attributes of types $E$ and $E_x$ in $P$ respectively, and 
$\circ \in \{>, \geq, <,$ $\leq, =, \neq\}$ be a comparison operator.
A type-grained trend count is maintained for a type $E$ if 
either there is no predicate of the form $(E.\mathit{attr} \circ E_x.\mathit{attr}_x)$ in $\theta$ or $E \notin P.\mathit{predTypes}(E_x)$. 
Otherwise, an event-grained trend count is computed for each matched event of type $E$.

Let $\mathcal{T}_t$ ($\mathcal{T}_e$) be the set of event types in $P$ for which type-grained (event-grained) trend counts are maintained.
Let $e \in I$ be an event of type $E$.
A mixed-grained trend count is computed as follows:
\[
\begin{array}{ll}
e.count = & \sum_{E' \in \mathcal{T}_t \cap P.\mathit{predTypes}(E)} E'.count +\\
& \sum_{e_p.\mathit{type} \in \mathcal{T}_e \cap P.\mathit{predTypes}(E),\;(e_p,e) \text{ satisfy } \theta} e_p.count.
\end{array}
\]
If $E=start(P)$, $e.count$ is incremented by one. $E.count$ and $final\_count$ are computed as defined in Theorem~\ref{theo:correctness-type}.
\label{theo:correctness-mixed}
\end{theorem}

Theorem~\ref{theo:correctness-mixed} can be proven by induction, similarly to Theorem~\ref{theo:correctness-type}. 


\begin{algorithm}[t]
\caption{Mixed-grained trend count algorithm}
\label{algo:mixed}
\begin{algorithmic}[1]
\Require Query $q$ with pattern $P$ and predicates $\theta$, stream $I$
\Ensure Count of event trends matched by $q$ in $I$
\State $H \leftarrow \text{empty hash table};\ V \leftarrow \emptyset;\ \mathit{final}\_\mathit{count} \leftarrow 0$
\ForAll {event type $E$ in $P$} 
	$H.put(E,0)$
\EndFor
\ForAll {$(E.\mathit{attr}\ Op\ E_x.\mathit{attr}_x) \in \theta$}
	\If {$E \in P.\mathit{predTypes}(E_x)$} 
    	$H.\mathit{remove}(E)$ 
    \EndIf
\EndFor
\ForAll {$e \in I$ of type $E$} 
	\State $e.\mathit{count} \leftarrow (E = \mathit{start}(P))\ ?\ 1\ :\ 0$
	\ForAll {$E' \in P.\mathit{predTypes}(E)$}  
    	\If {$E' \in H$} 
        	$\switch e.count += H.get(E')$ 	 		
        \Else \textbf{ for each } $e_p\in V.\mathit{predEvents}(e)$ of type $E'$
       			\State\ $|$ \ \ \ \textbf{do } $V \leftarrow V \cup e;\ \switch e.count += e_p.count$              
        \EndIf
  	\EndFor   	
    \If {$E \in H$} 
    	\State $E.count \leftarrow H.get(E) + e.count$
        \State $H.put(E, E.count)$ 		
    \Else \textbf{ if } $E = end(P)$ \textbf{ then } $\switch \mathit{final}\_\mathit{count} += e.count$ 
    \EndIf
\EndFor
\If {$end(P) \in H$} 
	\Return $H.get(end(P))$
	\Else\ \Return $\mathit{final}\_\mathit{count}$
\EndIf
\end{algorithmic}
\end{algorithm} 

\textbf{Mixed-Grained Trend Count Algorithm}.
At compile time, for each type $E$ in the pattern $P$, Algorithm~\ref{algo:mixed} decides whether to maintain a single type-grained trend count $E.count$ or an event-grained trend count for each matched event of type $E$. Type-grained trend counts are maintained in a hash table $H$. Initially, the table $H$ contains all types $P$ with counts set to 0 (Lines~1--2). For each predicate that restricts the adjacency relation between events of type $E$ and events of type $E_x$, if $E$ is a predecessor type of $E_x$, then $E$ is removed from the table $H$ since event-grained trend counts must be computed for events of type $E$ (Lines~3--4). 

At runtime, for each event $e$ of type $E$, $e.count=1$ if $E$ is a start type of $P$. Otherwise, $e.count=0$ (Lines~5--6). For each predecessor type $E'$ of $E$, if a type-grained trend count $E'.count$ is maintained, $e.count$ is incremented by $E'.count$ (Lines~7--8). If event-grained trend counts for events of type $E'$ are computed, $e.count$ is incremented by the count of each predecessor event $e'$ of type $E'$ (Lines~9--10).
If a type-grained trend count $E.count$ is maintained, $E.count$ is incremented by $e.count$ in the table $H$ (Lines~11--13). If event-grained trend counts for events of type $E$ are computed and $E$ is an end type of $P$, then the final count is incremented by $e.count$ (Line~14).
Lastly, if a type-grained trend count is maintained for the end type of $P$, it is returned as a result. Otherwise, $final\_count$ is returned (Lines~15--16).

\begin{theorem}[\textbf{\textit{Complexity}}]
Let $q$ be a query evaluated under skip-till-any-match,
$P$ be its pattern of length $l$, and
$n$ be the number of events per window of $q$.
Let $\mathcal{T}_t$ ($\mathcal{T}_e$) be the set of event types for which type-grained (event-grained) trend counts are maintained.
Let $t \leq l$ be the number of types in $\mathcal{T}_t$ and 
$n_e \leq n$ be the number of events of a type in $\mathcal{T}_e$ per window of $q$.
Algorithm~\ref{algo:mixed} has quadratic time $O(n(t+n_e))$ and linear space $\Theta(t + n_e)$ complexity.
\label{theo:comlexity-mixed}
\end{theorem}

\begin{proof}
The time complexity of the static analysis (Lines 1--4) is linear in $l$ and the number of predicates $\theta$. These values are negligible compared to the number of events $n$ for high-rate streams and meaningful queries.
During runtime execution (Lines~5--16), for each event, $O(t)$ type-grained and $O(n_e)$ event-grained trend counts are accessed. Thus, the time complexity is quadratic: $O(n(t+n_e))$.
The space complexity is determined by the number of trend counts: $\Theta(t+n_e)$.
\end{proof}

\begin{theorem}[\textbf{\textit{Time Optimality}}]
The quadratic time complexity $O(n(t+n_e))$ of Algorithm~\ref{algo:mixed} is optimal.
\label{theo:optimality-mixed}
\end{theorem}

\begin{proof}
By Theorem~\ref{theo:optimality-type}, the type-grained trend count computation has optimal time complexity $O(nt)$. Event-grained trend counts are maintained only for those events that are needed to verify their adjacency to future events.  In the worst case, each event triggers an access to $O(n_e)$ event-grained trend counts to guarantee correctness. In sum, the time complexity $O(n(t+n_e))$ is optimal.
\end{proof}

\section{Pattern-Grained Aggregator}
\label{sec:pattern-grained-aggregator}

To avoid \textit{polynomial time overhead} of trend construction under the skip-till-next-match and contiguous semantics (Table~\ref{tab:complexity}), we now propose to \textit{incrementally} compute trend aggregation at the coarsest possible \textit{pattern granularity}. This is possible because an event can have at most one predecessor event under these semantics. 

\begin{theorem}[\textbf{\textit{Predecessor Event Uniqueness}}]
Let $q$ be a query evaluated under the skip-till-next-match or contiguous semantics. 
For any event $e \in I$ that is part of an event trend matched by $q$, $e$ has the same predecessor event in all trends matched by $q$ (if $e$ has a predecessor).
\label{theo:predecessor-event-uniqueness}
\end{theorem}

\begin{proof}
Suppose $e$ has different predecessors $e_1$ and $e_2$ in trends $tr_1$ and $tr_2$, respectively. By Definition~\ref{def:adjacent_events}, $e_1.time < e.time$ and $e_2.time < e.time$. Without loss of generality, assume $e_1.time < e_2.time$. By Definitions~\ref{def:event-trend-under-CONT} and \ref{def:event-trend-under-STNM}, relevant events must be matched under the skip-till-next-match and contiguous semantics. Thus, the relevant event $e_2$ cannot be skipped in $tr_1$ and $e_1$ cannot be predecessor event of $e$.
\end{proof}

By Theorem~\ref{theo:predecessor-event-uniqueness} and Definition~\ref{def:adjacent_events}, event adjacency can be determined based on the last matched event and a new event. Thus, only last matched event must be stored.
%

\vspace{-2mm}
\begin{table}[h]
\centering
\begin{tabular}{|c||c|c|}
\hline
Event $e$
& $e.count$
& $\mathit{final}\_count$ \\
\hline
\hline
a1 & \textit{\textbf{1}} & \\ \hline
b2 & & \textit{\textbf{1}} \\ \hline
a3 & \textit{\textbf{2}} & \\ \hline 
a4 & \textit{\textbf{3}} & \\ \hline 
c5 & \textbf{3},\ \ \textit{0} & \\ \hline 
b6 & & \textbf{4},\ \ \textit{1} \\ \hline 
a7 & \textbf{4},\ \ \textit{1} & \\ \hline 
b8 & & \textbf{8},\ \ \textit{2} \\ \hline 
\end{tabular}
\vspace{2mm}
\caption{Pattern-grained trend count under the skip-till-next-match (bold) and contiguous (italics) semantics}
\label{tab:pattern}
\end{table}
\vspace{-6mm}

\begin{example}
Under skip-till-next-match, the trend count computation for our running example  is shown in bold in Table~\ref{tab:pattern}. 
Each $a$ increments the intermediate count $a.count$ by one since it is of a start type of $P$. 
Each $b$ increments the final count since it is of an end type of $P$. 
The irrelevant event $c5$ is skipped and eight trends are detected (Figure~\ref{fig:ess}).

Under the contiguous semantics, the trend count is computed analogously, shown in italics in Table~\ref{tab:pattern}. The only difference is that $c5$ cannot be ignored. To invalidate all partial trends that end at the last matched event $a4$, we set the last matched event to null and its intermediate count to 0. Since $b6$ is neither of a start type of $P$, nor adjacent to the last matched event (null), $b6$ cannot be matched. Two contiguous trends are detected (Figure~\ref{fig:ess}). 
\label{ex:pattern}
\end{example}

\begin{theorem}[\textbf{\textit{Pattern-Grained Trend Count}}]
Let $q$ be a query evaluated under the contiguous or skip-till-next-match semantics and $P$ be its pattern. Let $e_l, e \in I$ be events such that $e_l$ is the predecessor event of $e$ in a trend matched by $q$. 
Then, the pattern-grained counts are computed as follows.
\[
\begin{array}{ll}
e.\mathit{count} &= e_l.\mathit{count}.\\
\mathit{final}\_\mathit{count} &= \sum_{end.type = end(P)} end.\mathit{count}.
\end{array}
\]
If $E=start(P)$, $e.count$ is incremented by one.
\label{theo:correctness-pattern}
\end{theorem}

Theorem~\ref{theo:correctness-pattern} can be proven by induction on the number of matched events.

\begin{algorithm}[t]
\caption{Pattern-grained trend count algorithm}
\label{algo:pattern}
\begin{algorithmic}[1]
\Require Query $q$ with pattern $P$, event stream $I$
\Ensure Count of event trends matched by $q$ in $I$
\State $e_l \leftarrow null;\ e_l.count \leftarrow 0;\ \mathit{final}\_\mathit{count} \leftarrow 0$
\ForAll {$e \in I$ of type $E$} 
	\If {$\mathit{isMatched}(e_l,e)$}		
  		\State $e.count \leftarrow (E = \mathit{start}(P))\ ? \ 1\ :\ 0$  	
  		\If {$\mathit{isAdjacent}(e_l,e)$} 			
  			$\switch e.count += e_l.count$	      
        \EndIf
        \If {$E = \mathit{end}(P)$} 
        	$\switch \mathit{final}\_\mathit{count} += e.count$ 
  		\EndIf 
        \State $e_l \leftarrow e$ 
  	\Else \textbf{ if } $q.\mathit{semantics} = \mathit{\contsem}$ \textbf{ then}
  		 \State $\ \ |\ \ \ \ \ e_l \leftarrow null;\ e_l.count \leftarrow 0$
	\EndIf	
\EndFor
\State\Return $\mathit{final}\_\mathit{count}$
\end{algorithmic}
\end{algorithm}

\textbf{Pattern-Grained Trend Count Algorithm}. 
Initially, the last matched event $e_l$ is \textit{null}, $e_l.count$ and the final count are set to 0 (Line~1 in Algorithm~\ref{algo:pattern}).
An event $e$ is matched by the query $q$ if one of the following conditions holds:
(1)~If $e$ is of a start type of $P$, $e$ starts a new trend and $e.count$ is set to one (Line~4).
(2)~If $e_l$ and $e$ are adjacent, $e$ continues existing partial trends and $e.count$ is increased by $e_l.count$ 
to accumulate the number of previously started partial trends (Line~5).
If $e$ is of an end type of $P$, the final count is increased by $e.count$ (Line~6).
The event $e$ becomes the new last event (Line~7).
If $e$ is not matched by the query $q$ that detects contiguous trends, 
then the last matched event $e_l$ and its count are reset 
to invalidate partial trends that end at $e_l$ (Lines~8--9).
The final count is not reset, however, because it accumulates the number of contiguous trends that were detected before the irrelevant event $e$ arrived. The final count is returned (Line~10).

\begin{theorem}[\textbf{\textit{Complexity}}]
Let $q$ be a query evaluated under the contiguous or skip-till-next-match semantics and
$n$ be the number of events per window of $q$.
Algorithm~\ref{algo:pattern} has linear time $O(n)$ and constant space $O(1)$ complexity.
\label{theo:complexity-pattern}
\end{theorem}

\begin{proof}
The time complexity is determined by the number of matched events: $O(n)$. 
The space complexity is constant since two counts and the last matched event are stored.
\end{proof} 

\begin{theorem}[\textbf{\textit{Time Optimality}}]
The linear time complexity $O(n)$ of Algorithm~\ref{algo:pattern} is optimal.
\label{theo:optimality-pattern}
\end{theorem}

\begin{proof}
Any event trend aggregation algorithm must process $n$ events to ensure correctness of the final aggregate. Thus, the linear time complexity $O(n)$ is optimal. 
\end{proof}

\section{Other Query Clauses}
\label{sec:other}

So far we have focused on the Kleenes patterns, event matching semantics, and predicates on adjacent events. In this section, we handle sliding windows, predicates on single events, 
%
%
and grouping to support all query clauses (Definition~\ref{def:query}). Other possible extensions of the language are discussed in Section~\ref{sec:discussion}.

\textbf{Sliding Windows} partition the unbounded stream into finite time intervals (Definition~\ref{def:query}). Since these intervals may overlap, an event $e$ may fall into $k \geq 1$ windows. The event $e$ may expire in some windows but remain valid in others. To tackle such contiguous nature of streaming, we maintain the counts per window. Each count per event (or type) is assigned a window identifier $wid$~\cite{LMTPT05-2} and is computed based on the predecessor events (and/or types) with the same identifier $wid$. Analogously, the final count for window $wid$ is computed based on the counts of events of end type of $P$ (or count of the end type of $P$) for the window $wid$~\cite{PLRM18}.

\textbf{Predicates on Single Events} are classified into~\cite{PLRM18,QCRR14}: 

\textbf{\textit{Local predicates}} restrict the attribute values of matched events For example, the predicate \textit{(M.activity = passive)} in query $q_1$ in Section~\ref{sec:introduction} selects those measurements that were taken during passive physical activities. Such predicates filter the stream. Thereafter, our \app\ approach applies.

\textbf{\textit{Equivalence predicates}} require that all events in a trend to carry the same value of an attribute. For example, query $q_2$ has a predicate \textit{[driver]} that requires all events in an Uber pool trip to have the same driver identifier. Such predicates partition the stream into non-overlapping sub-streams by the values of this attribute enabling scalable parallel execution (Section~\ref{sec:discussion}). Thereafter, our \app\ approach applies to each sub-stream. 


\textbf{Event Trend Grouping} ensures that all events in a trend carry the same values of grouping attributes, e.g., queries $q_1$--$q_3$ in Section~\ref{sec:introduction}. Similarly to equivalence predicates, event trend grouping partitions the stream into non-overlapping sub-streams. Then, our approach applies to each sub-stream.

\section{Discussion}
\label{sec:discussion}

We now discuss how to extend our query language (Definition~\ref{def:query}) by other features to relax the assumptions in Section~\ref{sec:model}.

\begin{table}[t]
\centering
\begin{tabular}{|p{3.2cm}|p{5.5cm}|p{4.4cm}|}
\hline
\textbf{Type-grained}\newline\textbf{aggregates}
& \textbf{Mixed-grained}\newline\textbf{aggregates}
& \textbf{Pattern-grained}\newline\textbf{aggregates}
\\
\hline
$\textbf{\textit{e.count}}_\textbf{\textit{E}} =$ 
\newline
$e.count +$
\newline
$\sum_{E' \in \mathcal{T}} E'.count_E$
\newline\newline
$\textbf{\textit{x.count}}_\textbf{\textit{E}} =$
\newline
 $\sum_{X' \in \mathcal{Y}} X'.count_E$
& 
$\textbf{\textit{e.count}}_\textbf{\textit{E}} =$ 
\newline
$e.count +$
\newline
 $\sum_{E' \in \mathcal{T}_t} E'.count_E +$
 \newline
 $\sum_{e".type \in \mathcal{T}_e,\;(e",e)\theta=\top} e".count_E$
 \newline
$\textbf{\textit{x.count}}_\textbf{\textit{E}} =$
\newline
 $\sum_{X' \in \mathcal{Y}_t} X'.count_E +$
 \newline
 $\sum_{x".type \in \mathcal{Y}_e,\;(x",x)\theta=\top} x".count_E$
& 
$\switch \textbf{\textit{e}}_\textbf{\textit{l}}\textbf{\textit{.count}}_\textbf{\textit{E}} +=$ 
\newline
$e.count$
\\
\multicolumn{2}{|l|}{\hspace*{2cm}$\textbf{\textit{T.count}}_\textbf{\textit{E}} = \sum_{t.type = T} t.count_E$}
& 
$\textbf{\textit{COUNT(E)}} =$
\\
\multicolumn{2}{|l|}{\hspace*{2cm}$\textbf{\textit{COUNT(E)}} = end(P).count_E$}
& 
 $\sum_{e_{l}.type = end(P)} e_{l}.count_E$
\\
\hline
$\textbf{\textit{e.min}} =$ 
\newline
 $\textsf{min}_{E' \in \mathcal{T}}$ 
\newline
 $(e.attr,E'.min)$
\newline
$\textbf{\textit{x.min}} =$ 
\newline
 $\textsf{min}_{X' \in \mathcal{Y}}(X'.min)$
& 
$\textbf{\textit{e.min}} =$ 
\newline
 $\textsf{min}_{E' \in \mathcal{T}_t,\ e".type \in \mathcal{T}_e,\ (e",e)\theta = \top}$
 \newline
 $(e.attr,E'.min,e".min)$\newline
$\textbf{\textit{x.min}} =$ 
\newline
 $\textsf{min}_{X' \in \mathcal{Y}_t,\ x".type \in \mathcal{Y}_e,\ (x",x)\theta = \top}$
 \newline
 $(X'.min,x".min)$
& 
$\textbf{\textit{e}}_\textbf{\textit{l}}\textbf{\textit{.min}} =$ 
\newline
$\textsf{min}(e.attr,e_{l}.min)$
\\
\multicolumn{2}{|l|}{\hspace*{2cm}$\textbf{\textit{T.min}} = \textsf{min}_{t.type = T}(t.min)$}
& 
$\textbf{\textit{MIN(E.attr)}} =$
\\
\multicolumn{2}{|l|}{\hspace*{2cm}$\textbf{\textit{MIN(E.attr)}} = end(P).min$}
& 
 $\textsf{min}_{e_l.type=end(P)}(e_l.min)$
\\
\hline
$\textbf{\textit{e.sum}} =$ 
\newline
 $e.attr*e.count +$
\newline 
 $\sum_{E' \in \mathcal{T}} E'.sum$
\newline\newline
$\textbf{\textit{x.sum}} =$ 
\newline
 $\sum_{X' \in \mathcal{Y}} X'.sum$
& 
$\textbf{\textit{e.sum}} =$ 
\newline
 $e.attr*e.count +$
\newline 
 $\sum_{E' \in \mathcal{T}_t} E'.sum +$\newline
 $\sum_{e".type \in \mathcal{T}_e,\ (e",e)\theta = \top} e".sum$
\newline
$\textbf{\textit{x.sum}} =$ 
\newline
 $\sum_{X' \in \mathcal{Y}_t} X'.sum +$
\newline
 $\sum_{x".type \in \mathcal{Y}_e,\ (x",x)\theta = \top} x".sum$
& 
$\switch \textbf{\textit{e}}_\textbf{\textit{l}}\textbf{\textit{.sum}} +=$ 
\newline
$e.attr * e.count$
\\
\multicolumn{2}{|l|}{\hspace*{2cm}$\textbf{\textit{T.sum}} = \sum_{t.type = T} t.sum$}
& 
$\textbf{\textit{SUM(E.attr)}} =$
\\
\multicolumn{2}{|l|}{\hspace*{2cm}$\textbf{\textit{SUM(E.attr)}} = end(P).sum$}
& 
 $\sum_{e_l.type = end(P)} e_l.sum$
\\
\hline
\end{tabular}
\vspace{2mm}
\caption{Coarse-grained event trend aggregation 
($e,x,e",x",t \in I$ are matched events,
$e.type = E$, 
$x.type = X \neq E$, 
$T$ is any event type in $P$,
$P.predTypes(E) = \mathcal{T} = \mathcal{T}_t \mathbin{\dot{\cup}} \mathcal{T}_e$, 
$P.predTypes(X) = \mathcal{Y} = \mathcal{Y}_t \mathbin{\dot{\cup}} \mathcal{Y}_e$
where type-grained aggregates are maintained for types $\mathcal{T}_t$ and $\mathcal{Y}_t$) 
and event-grained aggregates are maintained for events of types  $\mathcal{T}_e$ and $\mathcal{Y}_e$}
\label{tab:aggregation-functions}
\end{table}

\textbf{Aggregation Functions}. 
While so far we have focused on event trend count computation, i.e., \mycount (*), we now sketch how the principles of coarse-grained online trend aggregation apply to other aggregation functions (Section~\ref{sec:model}). Table~\ref{tab:aggregation-functions} defines \mycount $(E)$, \mymin $(E.\mathit{attr})$, and\mysum $(E.\mathit{attr})$ at different granularities. In contrast to \mycount (*), only matched events $e$ of type $E$ update the aggregates. All other matched events $x$ of type $X \neq E$ propagate the aggregates from previous to more recent events in a trend matched by a pattern $P$ (or event types in $P$). \mymax (E.attr) is maintained analogously to \mymin $(E.\mathit{attr})$, while \myavg $(E.\mathit{attr})$ is computed based on \mysum $(E.\mathit{attr})$ and \mycount $(E)$ (Section~\ref{sec:model}).

\textbf{Negated Sub-Patterns}.
We split the pattern into positive and negative sub-patterns and maintain aggregates for each sub-pattern separately~\cite{PLRM18}. If aggregates are maintained per matched event, whenever a negative sub-pattern $N$ finds a match (i.e., its \mycount (*) = 1), all previously matched events of predecessor types $\mathcal{T}_p$ of $N$ are marked as incompatible with all future events of following types $\mathcal{T}_f$ of $N$. If aggregates are maintained per type, the aggregates of all predecessor types $\mathcal{T}_p$ are marked as invalid to contribute to aggregates of the following types $\mathcal{T}_f$. Lastly, if aggregates are maintained per pattern, $e_l$ of the previous sup-pattern of $N$ is set to null.

\textbf{Disjunction and Conjunction} can be supported by our \app\ approach without changing its
complexity because the aggregates for a disjunctive or a conjunctive pattern $P$ can be computed
based on the aggregates for the sub-patterns of $P$~\cite{PLRM18}.

\textbf{Kleene Star and Optional Sub-Patterns} can also be supported without changing the complexity
since they are syntactic sugar operators. Indeed, $\seq (P_i*, P_j) = \seq (P_i+, P_j) \vee P_j$ and $\seq (P_i?, P_j) = \seq (P_i, P_j ) \vee P_j$.

\textbf{Multiple Event Type Occurrences in a Pattern}.
If a type appears several time in a pattern, our \app\ approach applies with the following modifications. 
(1)~We assign an identifier to each type. For example, $\seq (A+, B, A)$ is translated into $\seq (A1+, B, A2)$. Then, each state in an FSA representation of the pattern has a unique label (Section~\ref{sec:pattern-analyzer}). 
(2)~An event $e$ may be inserted into several sub-graphs (or update several type-grained aggregates) under the skip-till-any-match semantics. However, e may not be its own predecessor event since an event may occur in a trend at most once~\cite{PLRM18}.

\textbf{Predicates on Minimal Trend Length} exclude too short and thus unreliable trends. One way of modeling such constraints is to unroll a pattern to its minimal length. For example, if we are interested trends matched by the pattern $A+$ and with $\mathit{length} \geq 3$, then we unroll the pattern $A+$ to length 3 as follows: $\textsf{SEQ}(A,A,A+)$. Our \app\ approach applies thereafter.

\textbf{Parallel Processing}.
As described in Section~\ref{sec:other}, equivalence predicates and \group\ clause partition the stream into sub-streams that are processed in parallel independently from each other. Such stream partitioning enables a highly scalable execution as demonstrated in Section~\ref{exp:grouping}. Analogously, each window is processed independently from other windows. However, events within the same substream and the same window must be processed in-order by their time stamp to guarantee correctness~\cite{sstore}. A stream transaction is a sequence of operations triggered by all events with the same time stamp. The application time stamp of a transaction (and all its operations) coincides with the application time stamp of the triggering events. For each time stamp $t$, our time-driven scheduler waits till the processing of all transactions with time stamps smaller than $t$ is completed. Then, the scheduler extracts all events with the time stamp $t$, wraps their processing into transactions, and submits them for execution.

\section{Performance Evaluation}
\label{sec:evaluation}

\subsection{Experimental Setup}
\label{sec:exp_setup}

\textbf{Infrastructure}. 
We have implemented our approach in Java with JRE 1.7.0\_25 running on Ubuntu 14.04 with 16-core 3.4GHz CPU and 128GB of RAM. We execute each experiment three times and report their average results here.

\textbf{Methodology}. 
We compare of our \app\ approach to Flink~\cite{flink}, SASE~\cite{ZDI14}, A-Seq~\cite{QCRR14}, and \greta~\cite{PLRM18} since they cover the spectrum of state-of-the-art event aggregation approaches (Table~\ref{tab:spectrum}). We run Flink on the same hardware as our platform. To achieve a fair comparison, we implemented SASE, A-Seq, and \greta\ on top of our platform. While Section~\ref{sec:related_work} is devoted to a detailed discussion of these approaches, we summarize them in Table~\ref{tab:methodology}.

\vspace*{-1mm}
\begin{table}[h]
\centering
\begin{tabular}{|c||c|c|c|c|c|c|}
\hline
\multirow{2}{*}{Approach}
& Kleene
& \multicolumn{3}{|c|}{Semantics}
& Predicates on
& Online trend
\\
& closure
& \anysem
& \nextsem
& \contsem
& adjacent events
& aggregation
\\
\hline
\hline
Flink 
& $-$
& +
& $-$
& +
& +
& $-$
\\
\hline
SASE
& +
& +
& +
& +
& +
& $-$
\\
\hline
\greta
& +
& +
& $-$
& $-$
& +
& $-$
\\
\hline
A-Seq 
& $-$
& +
& $-$
& $-$
& $-$
& $-$
\\
\hline
\textbf{\app}
& \textbf{+}
& \textbf{+}
& \textbf{+}
& \textbf{+}
& \textbf{+}
& \textbf{+}
\\
\hline
\end{tabular}
\vspace{1mm}
\caption{Expressive power of the event aggregation approaches}
\label{tab:methodology}
\end{table}
\vspace*{-4mm}

$\bullet$ \textbf{\textit{Flink}}~\cite{flink} is a popular open-source streaming platform that supports pattern matching. We express our queries using Flink operators such as event sequence, window, and grouping. Similarly to other industrial systems~\cite{esper,oracle}, Flink does not support Kleene closure. Thus, we flatten our queries as follows. 
For each Kleene pattern $P$, we first determine the length $l$ of the longest match of $P$. We then specify a set of fixed-length event sequence queries that cover all possible lengths up to $l$. Flink supports the skip-till-any-match and contiguous semantics. It implements a two-step approach that constructs all event sequences prior to their aggregation.

$\bullet$ \textbf{\textit{SASE}}~\cite{ZDI14} supports Kleene closure and all event matching semantics. It implements the two-step approach. Namely, it first stores each event $e$ in a stack and computes the pointers to $e$'s previous events in a trend. For each window, a DFS-based algorithm traverses these pointers to construct all trends. Then, these trends are aggregated. 

$\bullet$ \textbf{\textit{\greta}}~\cite{PLRM18} captures all matched events and the trend relationships among them as a graph. Based on the graph, it computes event trend aggregation online, that is, it avoids trend construction. However, it supports only skip-till-any-match semantics and maintains aggregates at a fine granularity. 

$\bullet$ \textbf{\textit{A-Seq}}~\cite{QCRR14} avoids event sequence construction by dynamically maintaining a count for each prefix of a pattern. Since A-Seq does not support Kleene closure, we flatten our queries as described above. A-Seq supports only the skip-till-any-match semantics. 
It does not support arbitrary predicates on adjacent events beyond equivalence predicates, such as \textit{[patient]} in query $q_1$ in Section~\ref{sec:introduction}.

\textbf{Data Sets}. 
We compare our \app\ approach to the state-of-the-art techniques using the following data sets.

$\bullet$ \textbf{\textit{Physical activity monitoring real data set}}~\cite{RS12} contains physical activity reports for 14 people during 1 hour 15 minutes. 18 activities are considered. A report carries time stamp in seconds, person identifier, activity identifier, and heart rate. The size of the data set is 1.6GB.   

$\bullet$ \textbf{\textit{Stock real data set}}~\cite{stockStream} contains 225k transaction records of 19 companies in 10 sectors. Each event carries time stamp in seconds,  company identifier, sector identifier, transaction identifier, transaction type (sell or buy), volume, price, etc. We replicate this data set 10 times. 

$\bullet$ \textbf{\textit{Public transportation synthetic data set}}. 
Our stream generator creates trips for 30 passengers using public transportation services in a city with 100 stations. Each event carries a time stamp in seconds, passenger identifier, station identifier, and waiting time in seconds. Waiting durations are generated uniformly at random.

\textbf{Event Queries}.
We evaluate variations of queries $q_1$--$q_3$ in Section~\ref{sec:introduction} against these data sets. We vary the event matching semantics, the number of events per window, predicate selectivity and the number of trend groups.
Unless stated otherwise, we evaluate our queries under the skip-till-any-match semantics since all state-of-the-art approaches support this semantics (Table~\ref{tab:methodology}).
To ensure that the two-step approaches terminate at least in some cases, the default number of events per window is 50k.
Since A-Seq does not support arbitrary predicates on adjacent events, we evaluate our queries without such predicates by default.
Unless stated otherwise, the number of trend groups is 
14 for the physical activity data set, 
19 for the stock data set, and 
30 for the public transportation data set.

\textbf{Metrics}. 
We measure the common metrics for streaming systems, namely, latency, throughput, and peak memory. 
\textbf{\textit{Latency}} is measured in milliseconds as the average time difference between the time of the aggregation result output and the arrival time of the latest event that contributes to the respective result.
\textbf{\textit{Throughput}} corresponds to the average number of events processed by all queries per second.
\textbf{\textit{Peak memory}} includes the memory for storing
the aggregates and sub-graphs for \app,
the \greta\ graph for \greta,
prefix counters for A-Seq,
events in stacks, pointers between them, and trends for SASE, and 
trends for Flink.

\subsection{Event Matching Semantics}

\begin{figure}[t]
\centering
\begin{minipage}{0.45\textwidth}
	\centering
    \includegraphics[width=0.7\textwidth]{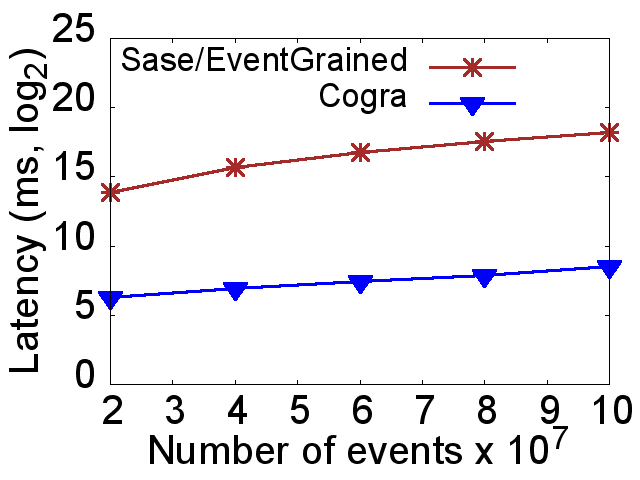}
    \vspace*{-2mm}
    \caption{Contiguous semantics (All approaches, physical activity real data set)}
    \label{fig:cont}
\end{minipage}
\hspace*{1cm}
\begin{minipage}{0.45\textwidth}
	\centering
    \includegraphics[width=0.7\textwidth]{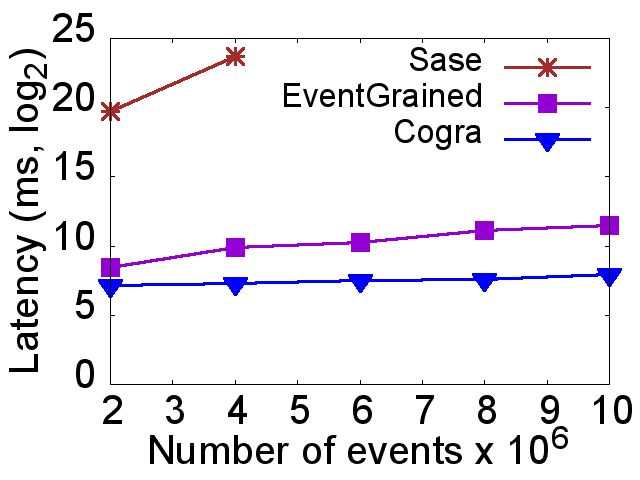}
    \vspace*{-2mm}
    \caption{Skip-till-next-match semantics (All approaches, public transportation data set)}
    \label{fig:next}
\end{minipage}
 \end{figure}

In Figures~\ref{fig:cont}--\ref{fig:any-online}, we compare the performance of \app\ to the state-of-the-art approaches under diverse event matching semantics, while varying the number of events per window. If an approach does not support a semantics (Table~\ref{tab:methodology}), the approach is not shown in the chart for this semantics.

\textbf{Two-step approaches} perform well for high-rate event streams only under the most restrictive contiguous semantics (Figure~\ref{fig:cont}) because the number and length of contiguous trends are relatively small. Nevertheless, \app\ achieves 27--fold speed up compared to Flink and 12--fold speed up compared to SASE when the number of events per window reaches 100M.

\textbf{\textit{Flink}} is not optimized for event trend aggregation for the following two reasons. 
(1)~Flink evaluates a workload of event sequence queries for each Kleene query. 
(2)~Flink first constructs all event sequences and then aggregates them. 
Consequently, under the skip-till-any-match semantics, the latency and memory usage of Flink grow exponentially in the number of events per window, while its throughput decreases exponentially (Figure~\ref{fig:any-all}). Flink does not terminate when the number of events exceeds 40k. For 40k events, \app\ achieves 4 orders of magnitude speed-up and uses 8 orders of magnitude less memory than Flink. 

\textbf{\textit{SASE}} supports Kleene patterns but also implements a two-step approach. Under the skip-till-any-match semantics, the latency and memory usage of SASE grow exponentially in the number of events, while its throughput degrades exponentially (Figure~\ref{fig:any-all}). SASE fails to terminate when the number of events exceeds 40k. For 40k events, \app\ achieves 3 orders of magnitude speed-up and 4 orders of magnitude memory reduction compared to SASE. 

Even under skip-till-next-match, SASE does not terminate if a window contains over 4M events (Figure~\ref{fig:next}). For 4M events, SASE returns results with an over 3 hours long delay -- which is unacceptable for time-critical applications. \app\ achieves 4 orders of magnitude speed-up and 5 orders of magnitude memory reduction compared to SASE in this case.



\begin{figure}[t]
    \centering
    \subfigure[Latency]{
    	\includegraphics[width=0.3\textwidth]{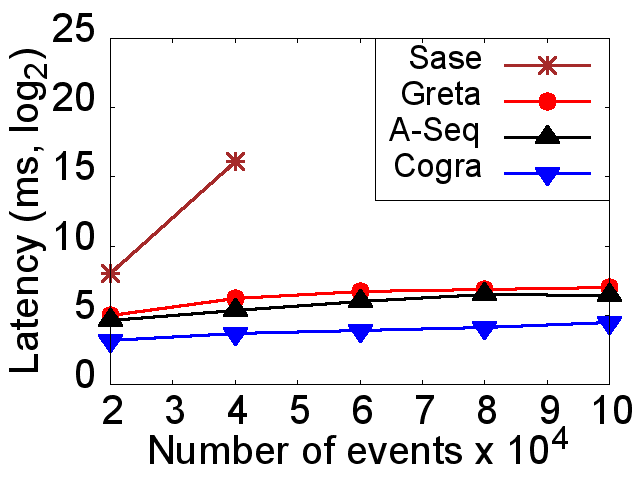}
    	 \label{fig:any-latency}
	}
	\subfigure[Memory]{
    	\includegraphics[width=0.3\textwidth]{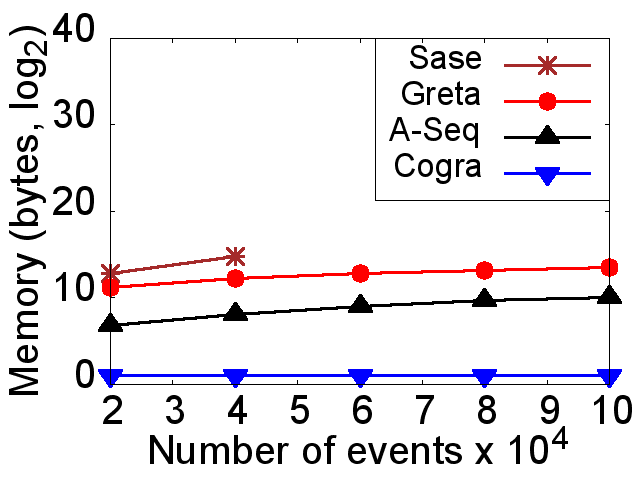}
    	\label{fig:any-mem}
	}
	\subfigure[Throughput]{
    	\includegraphics[width=0.3\textwidth]{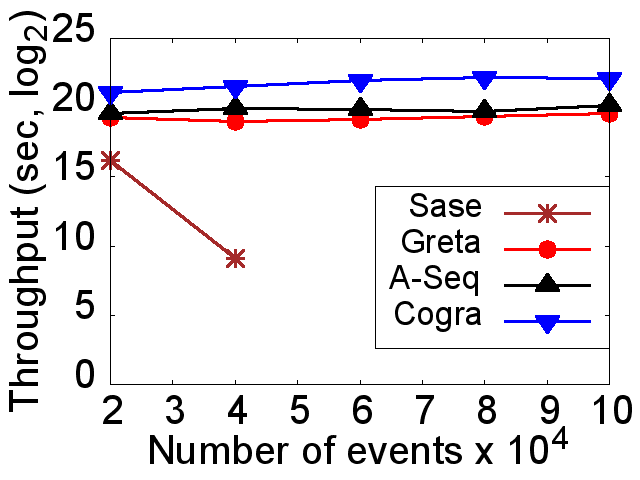}
    	 \label{fig:any-thru}
	}
    \vspace*{-2mm}
    \caption{Skip-till-any-match semantics (All approaches, stock real data set)}
    \label{fig:any-all}
\end{figure}
\begin{figure}[t]
    \centering
 	\subfigure[Latency]{
    	\includegraphics[width=0.3\textwidth]{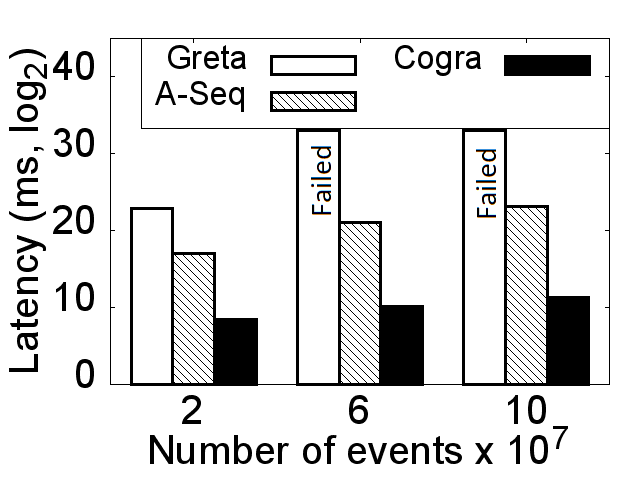}
    	\label{fig:any-latency-online}
	}
    \subfigure[Memory]{
    	\includegraphics[width=0.3\textwidth]{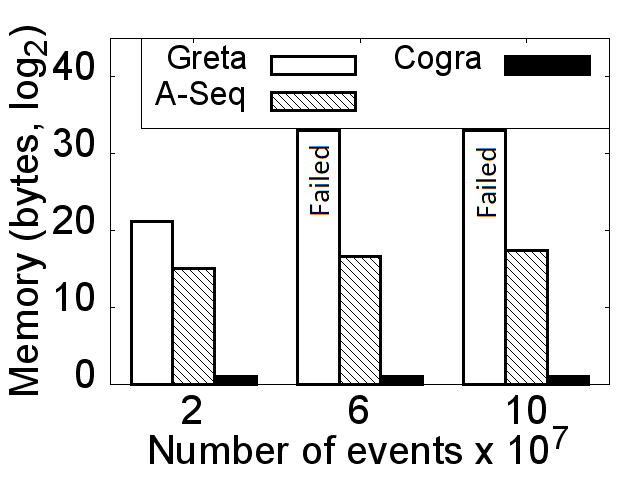}
    	\label{fig:any-mem-online}
	}
	\subfigure[Throughput]{
    	\includegraphics[width=0.3\textwidth]{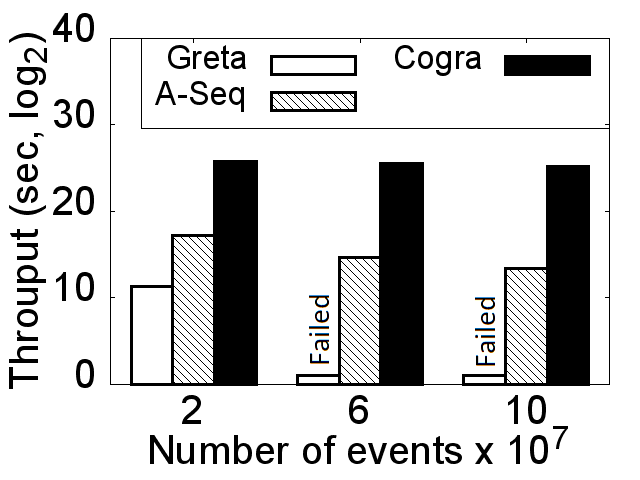}
    	 \label{fig:any-thru-online}
	}
    \vspace*{-2mm}
    \caption{Skip-till-any-match semantics (Online approaches, stock real data set)}
    \label{fig:any-online}
\end{figure}

\textbf{Online approaches} perform similarly for low-rate event streams (Figure~\ref{fig:any-all}), while high-rate streams reveal the difference between them (Figure~\ref{fig:any-online}).

\textbf{\textit{\greta}} captures all matched events and their trend relationships as a graph. While the overhead of graph construction is negligible for low-rate streams, it becomes a bottleneck when the stream rate increases. Under skip-till-any-match, \greta\ does not terminate if the stream rate is over 20M events per window (Figure~\ref{fig:any-latency-online}). For 20M events, \greta\ suffers from over an hour long delay which is 4 orders of magnitude higher compared to \app.

\textbf{\textit{A-Seq}} performs best among the state-of-the-art approaches. However, its expressive power is limited (Table~\ref{tab:methodology}). Also, A-Seq must evaluate a workload of event sequence queries for each Kleene pattern. The number of queries grows linearly in the number of events. Thus, the latency of A-Seq is 3 orders of magnitude higher compared to \app\ when the number of events per window reaches 100M (Figure~\ref{fig:any-latency-online}). 
Each event sequence query maintains a fixed number of aggregates~\cite{QCRR14}. Thus, the memory usage of A-Seq grows linearly with the number of queries (i.e., with the number of events). A-Seq requires 4 orders of magnitude more memory than \app\ for 100M events (Figure~\ref{fig:any-mem-online}).

\textbf{\textit{Cogra}} performs well for all semantics and stream rates because \app\ maintains a fixed number of aggregates per Kleene pattern. Its memory usage is constant (Figures~\ref{fig:any-mem} and \ref{fig:any-mem-online}).
The latency of \app\ grows linearly in the number of events. For 100M events per window, the latency of \app\ stays within 3 seconds (Figure~\ref{fig:any-latency-online}), while its throughput is over 39M events per second (Figure~\ref{fig:any-thru-online}).
In short, \app\ enables real-time and in-memory event trend aggregation.


\subsection{Predicate Selectivity}

In Figure~\ref{fig:predicates}, we focus on the selectivity of predicates on adjacent events, while predicates on single events determine the number of trend groups (Section~\ref{sec:other}) that is varied in Section~\ref{exp:grouping}.
To ensure that the two-step approaches terminate in most cases, we run this experiment against a low-rate stream of 50k events per window.
Since A-Seq does not support expressive predicates, it is not evaluated here.

\textbf{Two-step approaches}. 
When the predicate selectivity increases, more and longer trends are constructed and stored by \textbf{\textit{Flink}}. Since its latency and memory usage grow exponentially (Table~\ref{tab:complexity}, Figure~\ref{fig:predicates}), Flink fails to terminate once the predicate selectivity exceeds 50\%. When the predicate selectivity is 50\%, \app\ achieves 3 orders of magnitude speed-up and 3 orders of magnitude memory reduction compared to Flink. 

\textbf{\textit{SASE}} constructs all trends. Thus, its latency grows exponentially. Only the current trend is stored however. The number of stored pointers between events increases when the predicate selectivity grows. Thus, the memory costs grow linearly. The latency of \app\ is 2 orders of magnitude lower, while it uses 13--fold less memory than SASE for 90\% predicate selectivity.

 \begin{figure}[t]
 \centering
      \subfigure[Latency]{
     	\includegraphics[width=0.3\columnwidth]{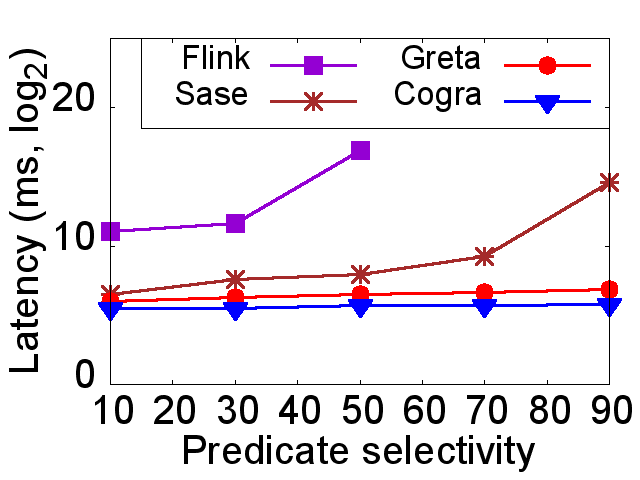}
     	 \label{fig:predicates-latency}
 	}
 	\hspace*{1cm}
 	\subfigure[Memory]{
     	\includegraphics[width=0.3\columnwidth]{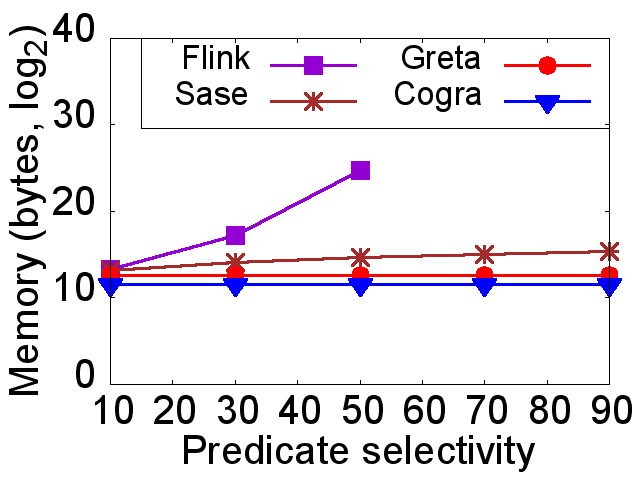}
     	\label{fig:predicates-mem}
 	}
     \vspace*{-2mm}
 	\caption{Predicate selectivity (All approaches, stock data set)} 
     \label{fig:predicates}
 \end{figure}

\textbf{Online approaches} perform well for such low-rate stream of 50k events per window. When the predicate selectivity increases, \textbf{\textit{\greta}} stores the same events in the \greta\ graph but the number of edges between them increases. Thus, latency increases linearly in the number of edges. Since edges are not stored, the memory consumption remains stable.

Similarly, the number of aggregates maintained by \textbf{\textit{\app}} stays the same with the growing predicate selectivity. Since \app\ maintains aggregates per event type (not per event), it achieves double speed-up and memory reduction compared to \greta\ when the predicate selectivity reaches 90\%. 

 \subsection{Event Trend Grouping}
 \label{exp:grouping}

In Figure~\ref{fig:groups}, we vary the number of trend groups. To ensure that the two-step approaches terminate in most cases, we run this experiment against a low-rate stream of 50k events per window. 
When there are only few groups, then the maintenance costs of each group can be costly due to high contents. Thus, the latency of all approaches reduces with the increasing number of trend groups.

\textbf{Two-step approaches}. 
Since trends are constructed per group, the number and length of trends decrease with the growing number of groups. Thus, the latency and memory costs of \textbf{\textit{Flink}} decrease exponentially when the number of groups increases. Flink fails to terminate when the number of groups is fewer than 15. For 15 groups, \app\ wins 5 orders of magnitude with respect to latency and 8 orders of magnitude regarding memory compared to Flink. 

\textbf{\textit{SASE}} constructs all trends without storing them. Thus, its latency reduces exponentially, while its memory consumption decreases linearly with the growing number of trend groups. SASE does not terminate for fewer than 25 groups. For 25 groups, \app\ achieves 4 orders of magnitude speed-up and 3 orders of magnitude memory reduction compared to SASE.

 \begin{figure}[t]
 	\centering
     \subfigure[Latency]{
     	\includegraphics[width=0.3\columnwidth]{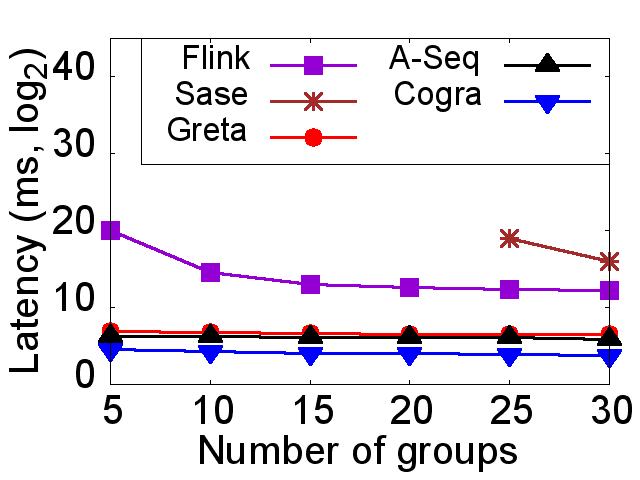}
     	 \label{fig:groups-latency}
 	}
 	\hspace*{1cm}
 	\subfigure[Memory]{
     	\includegraphics[width=0.3\columnwidth]{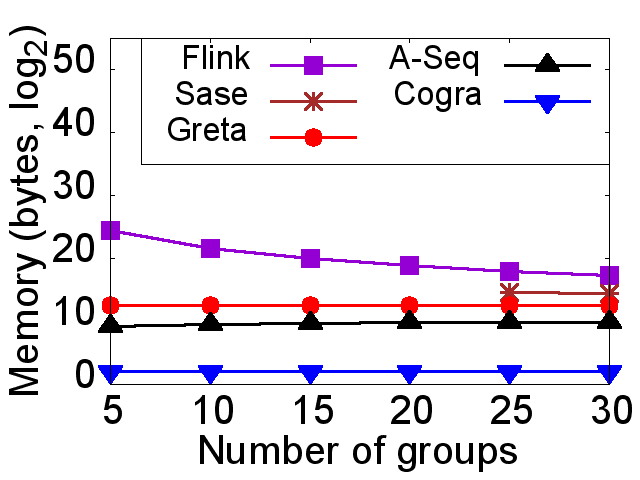}
     	\label{fig:groups-mem}
 	}
     \vspace*{-2mm}
     \caption{Event trend grouping (All approaches, public transportation data set)}
     \label{fig:groups}
 \end{figure}

\textbf{Online approaches} perform well independently from the number of event trend groups since trends are not constructed. 
The latency of \textbf{\textit{A-Seq}} reduces linearly when the number of groups increases. Since A-Seq evaluates a workload of event sequence queries for each Kleene query, its latency is 5--fold higher than the latency of \app\ for 5 groups.
A-Seq maintains aggregates per group. Thus, its memory costs increase linearly in the number of groups. \app\ wins 2 orders of magnitude regarding memory usage compared to A-Seq for 30 groups.

The latency of \textbf{\textit{\greta}} decreases linearly with the increasing number of groups. Due to the \greta\ graph construction overhead, the latency of \greta\ is 7--fold higher than the latency of \app\ for 5 groups. 
The memory usage of \greta\ remains stable when varying the number of groups because the same number of events is stored in the \greta\ graphs. 
The memory costs of \greta\ are 3 orders of magnitude higher compared to \app\ in all cases.

\section{Related Work}
\label{sec:related_work}

\textbf{Complex Event Processing} approaches including SA\-SE~\cite{ADGI08, WDR06, ZDI14}, Cayuga~\cite{DGPRSW07}, ZStream~\cite{MM09}, AFA~\cite{CGM10}, E-Cube~\cite{LRGGWAM11}, and CET~\cite{PLAR17} solve orthogonal problems. 
Many of them deploy a Finite State Automaton-based query execution paradigm~\cite{ADGI08, WDR06, ZDI14, DGPRSW07, CGM10, MZZ12, DGNSZ15}. 
AFA supports dynamic pattern detection over disordered streams. 
ZStream translates an event query into an operator tree optimized using rewrite rules.
E-Cube employs hierarchical event stacks to share events across different event queries.
CET solves the trade-off between the CPU and memory during event trend detection.
Some approaches focus on XPath-based query processing over XML streams~\cite{MZZ12, DGNSZ15}.
Others implement hardware-based CEP on FPGAs at gigabit wire speed~\cite{WTA10}.
However, the expressive power of these approaches is limited. 
Indeed, they do not support 
Kleene closure~\cite{LRGGWAM11}, 
nor aggregation~\cite{PLAR17, CGM10, WTA10, MZZ12}, 
nor various event matching semantics~\cite{LRGGWAM11, PLAR17, DGPRSW07, MM09, CGM10, WTA10, MZZ12, DGNSZ15}.
In contrast to them, SASE supports all above query language constructs.
However, it does not design any optimization techniques for event trend aggregation and thus requires trend construction prior to trend aggregation. Due to the exponential time complexity of trend construction (Table~\ref{tab:complexity}), this two-step approach fails to respond within a few seconds (Section~\ref{sec:evaluation}).

In contrast, A-Seq~\cite{QCRR14} proposes \textit{online aggregation of fixed-length event sequences} under skip-till-any-match. However, it supports neither Kleene closure, nor other event matching semantics, nor expressive predicates on adjacent events. Thus, A-Seq does not tackle the challenges of this work.
\greta~\cite{PLRM18} defines \textit{online event trend aggregation} under skip-till-any-match. \greta\ does not support other semantics. It captures all event trends and aggregates as a graph. It avoids event trend construction and thus reduces the time complexity from exponential to quadratic in the number of events in the worst case. Since \greta\ maintains aggregates at the event granularity, it does not achieve optimal time complexity for several query classes and thus suffers from long delays  (Section~\ref{sec:evaluation}). 

\textbf{Traditional Data Streaming} approaches~\cite{AW04, GHMAE07, KWF06, LMTPT05, LMTPT05-2, THSW15, ZKOS05, ZKOSZ10} support aggregation computation over data streams. Some incrementally aggregate \textit{raw input events for single-stream queries}~\cite{LMTPT05, LMTPT05-2}. Others share aggregation results among overlapping sliding windows~\cite{AW04, LMTPT05} or multiple queries~\cite{KWF06, ZKOS05, ZKOSZ10}.
However, these approaches are restricted to Select-Project-Join queries with window semantics. That is, their execution paradigm is set-based. They support neither various event matching semantics nor CEP-specific operators such as event sequence and Kleene closure that treat the order of events as  a first-class citizen. These approaches require the \textit{construction of join results} prior to their aggregation, i.e., they define incremental aggregation of \textit{single raw events} but implement a two-step approach for join results.

Industrial streaming systems such as Flink~\cite{flink}, Esper~\cite{esper}, and Oracle Stream Analytics~\cite{oracle} support fixed-length event sequences. They support neither Kleene closure nor different semantics. While Kleene closure computation can be simulated by a set of event sequence queries covering all possible lengths of event trends, this approach is possible only if the maximal length of a trend is known apriori. This is rarely the case in practice. Furthermore, this approach is highly inefficient for two reasons. One, it must execute a set of event sequence queries for each Kleene query. This increased workload degrades system performance. Two, since this approach requires trend construction prior to their aggregation, it has exponential time complexity in the worst case (Table~\ref{tab:complexity}).

\textbf{Static Sequence Databases} extend traditional SQL queries by order-aware join operations and support aggregation of their results~\cite{LS03, LKHLCC08}. However, they do not support Kleene closure. Instead, \textit{single data items} are aggregated~\cite{LS03, MZ97, SZZA04, SLR96}. They also do not support various matching semantics. Lastly, these approaches assume that the data is statically stored and indexed prior to processing. Hence, they do not tackle challenges that arise due to dynamically streaming data such as event expiration and real-time processing. 

\textbf{Sequential Pattern Mining} aims at detecting top-k item sequences that frequently occur in a given data set. The sequence patterns of interest are not known prior to processing. These approaches employ a Prefix tree and the Apriori algorithm to optimize the detection~\cite{HPMCDH00, HPYM04, MDH08, MTV97, PHMPCDH01}. 
This problem is distinct from our problem in which the event patterns of interest are specified by the user prior to processing. Frequent pattern mining approaches do not tackle real-time aggregation of event trends detected by Kleene patterns under various semantics. However, our approach can be used to count the number of sequences and then select the top-k frequent of them.

\section{Conclusions}
\label{sec:conclusions}

Our \app\ approach is the first to aggregate Kleene pattern matches under various event matching semantics with optimal time complexity. To this end, \app\ incrementally maintains trend aggregates at coarse  granularity. Thus, it achieves up to four orders of magnitude speed-up and up to eight orders of magnitude memory reduction compared to state-of-the-art.

\bibliographystyle{abbrv}
\bibliography{cogra-tr}

\end{document}